\RequirePackage{rotating}
\documentclass{amsart}
\usepackage[utf8]{inputenc}
\usepackage[top=1in, bottom=1.5in, left=1.5in, right=1.5in]{geometry}

\usepackage[dvipsnames]{xcolor}
\definecolor{TUMblue}{HTML}{0065bd}
\usepackage[colorlinks=true,linkcolor=TUMblue,citecolor=OliveGreen,urlcolor=Sepia,linktocpage=true]{hyperref}
\usepackage{tgtermes}

\usepackage{amsaddr}

\usepackage{amsmath}
\usepackage{amssymb}
\usepackage{mathtools}
\usepackage{pifont}
\usepackage{braket}
\usepackage{amsthm}
\usepackage{upgreek}
\usepackage{comment}
\usepackage{stmaryrd}
\usepackage{appendix}
\usepackage{bbm}
\usepackage{dsfont}
\usepackage{url}
\usepackage{tikz-cd}
\usepackage[shortlabels]{enumitem}
\usepackage{caption}
\usepackage{subcaption}
\usepackage[shortcuts]{extdash}
\usepackage{rotating}

\DeclarePairedDelimiter{\floor}{\lfloor}{\rfloor}

\newtheorem{theorem}{Theorem}[section]
\newtheorem{corollary}[theorem]{Corollary}
\newtheorem{lemma}[theorem]{Lemma}
\newtheorem{proposition}[theorem]{Proposition}

\newtheorem{definition}[theorem]{Definition}
\newtheorem{remark}[theorem]{Remark}

\newcommand{\generate}[2]{\langle #1 \rangle_{\mathsf{#2}}}

\newcommand{\iu}{\mathrm{i}\mkern1mu}

\newcommand{\C}{\mathbb{C}}
\newcommand{\R}{\mathbb{R}}

\newcommand{\Z}{\mathbb{Z}}

\newcommand{\mf}[1]{\mathfrak{#1}}

\newcommand{\mb}[1]{\mathbf{#1}}

\newcommand{\down}{\shortdownarrow}
\newcommand{\plusdown}{{\mathrel{\ooalign{$\hss\shortdownarrow$\hss\cr\hss\raisebox{2pt}{\scalebox{0.6}{$-$}}\hss}}}}

\newcommand{\diag}{\operatorname{diag}}
\newcommand{\qdiag}{\operatorname{qdiag}}

\newcommand{\conv}{\operatorname{conv}}

\newcommand{\SU}{\operatorname{SU}}
\newcommand{\SO}{\operatorname{SO}}

\newcommand{\U}{\operatorname{U}}

\renewcommand{\Im}{\operatorname{Im}}
\renewcommand{\Re}{\operatorname{Re}}

\newcommand{\Ad}{\operatorname{Ad}}
\newcommand{\ad}{\operatorname{ad}}
\newcommand{\tr}{\operatorname{tr}}
\renewcommand{\det}{\operatorname{det}}

\newcommand{\id}{\mathds1}

\newcommand{\reach}{\mathsf{reach}}

\newcommand{\derv}{\mathsf{derv}}

\renewcommand{\epsilon}{\varepsilon}

\newcommand{\ketbrax}[2]{\ket{#1}\!\bra{#2}}

\newcommand{\dmin}{{d_{\min}}}
\newcommand{\Uloc}{\U_{\mathrm{loc}}}
\newcommand{\SUloc}{\mathrm{SU}_{\mathrm{loc}}}
\newcommand{\Uswap}{U_{\mathrm{swap}}}
\newcommand{\uloc}{\mf u_{\mathrm{loc}}}
\newcommand{\suloc}{\mf{su}_{\mathrm{loc}}}

\newcommand{\vecs}{\mf{H}}
\newcommand{\mad}{\ad^d}

\newcommand{\AIII}{\mathsf{AIII}}
\newcommand{\CI}{\mathsf{CI}}
\newcommand{\DIII}{\mathsf{DIII}}

\begin{document}

\title[Entanglement in Bipartite Systems with Local Unitary Control]{Entanglement in Bipartite Quantum Systems with Fast Local Unitary Control}

\author{Emanuel Malvetti}
\address{School of Natural Sciences, Technische Universit\"at M\"unchen, 85737 Garching, Germany, and Munich Center for Quantum Science and Technology (MCQST) \& Munich Quantum Valley (MQV)}

\begin{abstract}
The well-known Schmidt decomposition, or equivalently, the complex singular value decomposition, states that a pure quantum state of a bipartite system can always be brought into a ``diagonal'' form using local unitary transformations. 
In this work we consider a finite-dimensional closed bipartite system with fast local unitary control. 
In this setting one can define a reduced control system on the singular values of the state which is equivalent to the original control system. 
We explicitly describe this reduced control system and prove equivalence to the original system.
Moreover, using the reduced control system, we prove that the original system is controllable and stabilizable and we deduce quantum speed limits.
We also treat the fermionic and bosonic cases in parallel, which are related to the Autonne--Takagi and Hua factorization respectively. \bigskip

 \medskip

\end{abstract}

\maketitle



\section{Introduction}

\subsection{Motivation}

Entanglement is one of the distinguishing features of quantum mechanics, and it lies at the core of emerging quantum technologies. 
Generating sufficient entanglement is a prerequisite for pure state quantum computation~\cite{Vidal03}, it is at the root of quantum cryptography~\cite{BB84reprint,Gisin02} and it can act as a resource for increasing the sensitivity of quantum sensors~\cite{Degen17}.
At the same time many basic questions about entanglement remain unanswered, and our theoretical understanding of multipartite and mixed state entanglement is limited~\cite{Bengtsson17}.
One setting in which entanglement is well-understood is the case of a bipartite quantum system in a pure state, where the singular value decomposition (a.k.a.\ the Schmidt decomposition) can be employed.
Similar tools can also be applied to indistinguishable subsystems~\cite{Huckleberry13} (bosons and fermions), where the singular value decomposition is to be replaced with the Autonne--Takagi and Hua factorization respectively. 

In this paper we apply quantum control theory~\cite{DiHeGAMM08,dAless21} to closed bipartite quantum systems subject to fast local unitary control.
Factoring out the fast controllable degrees of freedom leads to a reduced control system defined on the singular values of the state.
Such reduced control systems have been addressed in~\cite[Ch.~22]{Agrachev04} under simplifying assumptions.
In particular, there, the reduced state space is assumed to be a smooth manifold without singularities.
In our case singularities occur where two singular values coincide or one singular value vanishes, and in fact they constitute the main technical difficulty.
In~\cite{Reduced23} this assumption is removed such that we can apply the results therein to the present setting.

Similar ideas have been applied to open Markovian quantum systems with fast unitary control in~\cite{Sklarz04,Yuan10,rooney2018} and by the author et al.\ in~\cite{LindbladReduced23}. 
There, the reduced control system describes the evolution of the eigenvalues of the density matrix (a positive semidefinite matrix of unit trace) representing a mixed quantum state. 
In quantum thermodynamics, a natural simplification of this system has been explored in~\cite{CDC19,vE_PhD_2020,MTNS2020_1,OSID23}.

Finally, in an upcoming paper~\cite{BipartiteOptimal} we will apply the methods derived here combined with optimal control theory to derive explicit control solutions for low dimensional systems.

\subsection{Outline}

We start by defining the bilinear control system describing a closed bipartite quantum system with fast local unitary control in Section~\ref{sec:bilinear-system}. 
We consider the case of distinguishable subsystems, as well as that of bosonic and fermionic subsystems. 
In Section~\ref{sec:matrix-decs} we show that these systems are related to certain matrix diagonalizations (which themselves are related to symmetric Lie algebras) and we derive corresponding reduced control systems in Section~\ref{sec:reduced-system}.
In Section~\ref{sec:global-phase} we show how some controllability assumptions can be weakened if one neglects the global phase of the quantum state.
Then we turn to some applications.
In Section~\ref{sec:ctrl-stab} we use the reduced control system to prove that the full bilinear control system is always controllable and stabilizable.
In Section~\ref{sec:speed-limit} we derive some quantum speed limits for the evolution of the singular values.
The connection between the present quantum setting and symmetric Lie algebras is drawn in Appendix~\ref{app:sym-lie-alg}, and this allows us to apply the results from~\cite{Reduced23}.

\section{Control Systems} 

First we define what we mean by a closed bipartite quantum system with fast local unitary control before we derive the equivalent reduced control system obtained by factoring out the local unitary action.
The reduced control system describes the dynamics of the singular values of the state and due to the normalization of the quantum state the reduced state space turns out to be a hypersphere.
The main result of this section is the equivalence, in a precise sense which will be made clear, of the full bilinear control system and the reduced one.
In particular there is no loss of information incurred by passing to the reduced control system.
A brief discussion of global phases concludes the section.

\subsection{Bilinear Control System} \label{sec:bilinear-system}

Let two finite dimensional Hilbert spaces $\C^{d_1}$ and $\C^{d_2}$ of dimensions $d_1,d_2\geq2$ representing the subsystems be given.
The total Hilbert space of the bipartite system is then $\C^{d_1}\otimes\C^{d_2}$.%
\footnote{Abstractly the symbol $\otimes$ denotes the tensor product, but since we always work with concrete vectors and matrices we interpret $\otimes$ as the Kronecker product. This identifies the vector spaces $\C^{d_1}\otimes\C^{d_2}$ and $\C^{d_1d_2}$ and similarly for matrices.}
We denote by $\mf{u}(d)$ the unitary Lie algebra consisting of skew-Hermitian matrices in $d$ dimensions.
Our goal is to study the full bilinear control system~\cite{Jurdjevic97,Elliott09} defined by the following controlled Schrödinger equation%
\footnote{Throughout the paper we set $\hbar=1$ and thus write the (uncontrolled) Schrödinger equation as $\ket{\dot\psi(t)}=-\iu H_0\ket{\psi(t)}$.} 
on $\C^{d_1}\otimes\C^{d_2}$:
\begin{equation} \label{eq:bilinear} \tag{\sf B} 
\ket{\dot\psi(t)} = -\iu\Big(H_0 
+ \sum_{i=1}^{m_1} u_{i}(t) E_{i}\otimes\id 
+ \sum_{j=1}^{m_2} v_{j}(t) \id\otimes F_{j}
\Big)\ket{\psi(t)},
\end{equation}
where $H_0 \in \iu\mf{u}(d_1)\otimes\iu\mf{u}(d_2) \cong \iu\mf{u}(d_1d_2)$ is the \emph{drift Hamiltonian} (or \emph{coupling Hamiltonian}), $E_{i}\in\iu\mf{u}(d_1)$ and $F_{j}\in\iu\mf{u}(d_2)$ are the \emph{control Hamiltonians}, and $u_{i}$ and $v_{j}$ are the corresponding \emph{control functions}. 
We make the following key assumptions:
\begin{enumerate}[(I)]
\item\label{it:fast-control} The control functions $u_i$ and $v_j$ are locally integrable, in particular they may be unbounded.
\item\label{it:full-control} The control Hamiltonians generate the full local unitary Lie algebra:
$$\generate{\iu E_{i}\otimes\id,\, \id\otimes\iu F_{j} : \,i=1,\ldots,m_1,\,j=1,\ldots,m_2}{\mathsf{Lie}}=\uloc(d_1,d_2),
$$
\end{enumerate}
where $\uloc(d_1,d_2):=(\mf u(d_1)\otimes\id)+(\id\otimes\,\mf u(d_2))$ is the Lie algebra of the Lie group $\Uloc(d_1,d_2):=\U(d_1)\otimes\U(d_2)$ of local unitary transformations.
Put simply, we have fast control over $\Uloc(d_1,d_2)$.

\begin{remark} \label{rmk:control-group}
One may also define the group $\SUloc(d_1,d_2):=\mathrm S(\U(d_1)\otimes\U(d_2))$ of \emph{local special unitary operations}.
Consider the local unitary $U=e^{\iu\phi_1}\id\otimes e^{\iu\phi_2}\id$. 
Then $\det(U)=e^{\iu d_1d_2(\phi_1+\phi_2)}$ and hence, if $U\in\SUloc$, the value of the applied phase $e^{\iu(\phi_1+\phi_2)}$ is restricted to a discrete set.
Thus, if we do not neglect the global phase of the state $\ket\psi$, fast control over $\SUloc(d_1,d_2)$ is not sufficient to generate all local unitary state transfers.
To simplify the exposition we assume fast control over $\Uloc(d_1,d_2)$, but we will revisit this issue in Section~\ref{sec:global-phase} to show how this assumption can be weakened.
\end{remark}

This covers the case of two distinguishable subsystems. 
However, we also wish to treat systems composed of two indistinguishable subsystems. 
In this case both subsystems have the same dimension $d:=d_1=d_2$. 
In the \emph{bosonic} case, the state $\ket\psi$ is unchanged by swapping the two subsystems, i.e., $\Uswap\ket\psi=\ket\psi$, where $\Uswap\ket{\psi_1}\otimes\ket{\psi_2} = \ket{\psi_2}\otimes\ket{\psi_1}$.
These ``symmetric'' states lie in the space $\mathrm{Sym}^2(\C^d)$.
In the \emph{fermionic} case, swapping yields a phase factor of $-1$, i.e., $\Uswap\ket\psi=-\ket\psi$. 
Such ``skew-symmetric'' states are contained in the space $\bigwedge^2(\C^d)$. 
In both cases the set of local unitaries applicable to the system is restricted to symmetric local unitaries $\Uloc^s(d):=\{V\otimes V:V\in\U(d)\}$. 
The corresponding Lie algebra is $\uloc^s(d):=\{\iu E\otimes\id + \id\otimes\iu E: \iu E\in\mf u(d)\}$ and is isomorphic to $\mf u(d)$.\footnote{Note however that the map $\U(d)\to\Uloc^s(d)$ given by $V\mapsto V\otimes V$ is a double cover with kernel $\{\id,-\id\}$.}
The set of all coupling Hamiltonians applicable to such systems is the set
$\mf u^s(d^2) = \{\iu H\in \mf u(d^2) : \Uswap H \Uswap^*=H \}$.
Hence, in the case of two indistinguishable subsystems the bilinear control system takes the form:
\begin{equation} \label{eq:bilinear2} \tag{\sf B'}  
\ket{\dot\psi(t)} = -\iu\Big(H_0
+ \sum_{i=1}^{m} u_{i}(t) (E_i\otimes\id_2+\id_1\otimes E_i)
\Big)\ket{\psi(t)}\,,
\end{equation}
where $H_0\in\iu \mf u^s(d^2)$ and $E_1,\dots,E_m\in\iu\mf u(d)$. 
Assumption~\ref{it:fast-control} remains unchanged, but Assumption~\ref{it:full-control} is slightly modified to state:
\begin{enumerate}[(I),resume]
\item \label{it:full-control-2} The local control Hamiltonians $\iu E_i\otimes\id_2+\id_1\otimes\iu E_i$ for $i=1,\ldots,m$ generate the full Lie algebra $\uloc(d)$.
\end{enumerate}
Note that the only difference between the bosonic and fermionic case is that the initial state of the control system~\eqref{eq:bilinear2} lies in $\mathrm{Sym}^2(\C^d)$ and $\bigwedge^2(\C^d)$ respectively.

\begin{remark} \label{rmk:control-group-sym}
Similarly to Remark~\ref{rmk:control-group}, one might consider control Hamiltonians of the form $E\otimes\id + \id\otimes E$ with the additional restriction of $\tr(E)=0$, defining the Lie algebra $\suloc^s(d)$.
In this case we again lose control over the global phase of the state.
Thus, for simplicity, we do not make this assumption here and refer to Section~\ref{sec:global-phase} for more details.
\end{remark}

\subsection{Related Matrix Decompositions} \label{sec:matrix-decs}
 
In order to derive and understand the reduced control system (which we will introduce in the next section) obtained by factoring out the local unitary action, we must first understand the mathematical structure of this action.
In the case of distinguishable subsystems, the local unitary action corresponds the complex singular value decomposition, an thus the only invariants of a state under this action are its singular values, which therefore are the natural choice for the reduced state.
The bosonic and fermionic cases correspond to less well-known matrix decompositions called the Autonne--Takagi and Hua factorization respectively.
Importantly all of these matrix decompositions also correspond to certain symmetric Lie algebras, and this is the key to applying the results on reduced control systems from~\cite{Reduced23} to the full bilinear control systems~\eqref{eq:bilinear} and~\eqref{eq:bilinear2}.

Let $\{\ket{i}_1\}_{i=1}^{d_1}$ and $\{\ket{j}_2\}_{j=1}^{d_2}$ denote the standard orthonormal bases%
\footnote{In the following we will usually omit the index denoting the subsystem, since it will be clear form the order.} 
of $\C^{d_1}$ and $\C^{d_2}$ respectively, and let $\ket{\psi}\in\C^{d_1}\otimes\C^{d_2}$ be a state vector.
The components $(\psi_{ij})_{i,j=1}^{d_1,d_2}$ of $\ket{\psi}$ are uniquely given by 
$$
\ket{\psi}=\sum_{i,j=1}^{d_1,d_2}\psi_{ij}\ket{i}_1\otimes\ket{j}_2=:\sum_{i,j=1}^{d_1,d_2}\psi_{ij}\ket{ij}.
$$ 
Hence every bipartite state can be represented by a matrix.%
\footnote{Our convention is consistent with~\cite[Sec.~9.2]{Bengtsson17} and~\cite{Huckleberry13}. 
In other contexts one often defines the vectorization operation $\mathrm{vec}(\cdot)$ which turns a matrix into a vector by stacking its columns and satisfies $\mathrm{vec}(AXB)=(B^\top\otimes A)\mathrm{vec}(X)$. 
Our convention is slightly different in that, identifying $\C^{d_1}\otimes\C^{d_2}\cong\C^{d_1d_2}$ via the Kronecker product, we may write $\ket\psi=\mathrm{vec}(\psi^\top)$.}
More precisely, we have used the canonical isomorphism  
\begin{align} \label{eq:matrix-iso}
\C^{d_1}\otimes\C^{d_2} \to \C^{d_1}\otimes(\C^{d_2})'\cong\C^{d_1,d_2}
\quad \ket{i}_1\ket{j}_2\mapsto\ket{i}_1\!\bra{j}_2,
\end{align}
where $(\cdot)'$ denotes the dual space.
We will use $\psi\in\C^{d_1,d_2}$ to denote the matrix corresponding to $\ket{\psi}$ under this isomorphism and vice versa.
For distinguishable subsystems, the matrix representing $\psi$ is an arbitrary complex matrix in $\C^{d_1,d_2}$ (the constraint induced by the normalization of the state $\ket\psi$ will be discussed in Remark~\ref{rmk:sphere} below). 
For indistinguishable subsystems, it holds that $d:=d_1=d_2$,
and the matrix $\psi\in\mf{sym}(\C,d)$ is symmetric in the bosonic case and $\psi\in\mf{asym}(\C,d)$ is skew-symmetric in the fermionic case.

Let $V\otimes W\in\Uloc(d_1,d_2)$ be a local unitary, and set $\ket{\phi}=V\otimes W\ket{\psi}$.
The coordinates of $V$ are then defined by $V=\sum_{k,i=1}^{d_1}V_{ki}\ketbrax{k}{i}$, and similarly $W=\sum_{l,j=1}^{d_2}W_{lj}\ketbrax{l}{j}$.
Then $\phi_{kl}=\sum_{i,j=1}^{d_1,d_2}V_{ki} W_{lj} \psi_{ij}$.
In matrix form this can be rewritten as $\phi=V\psi W^\top$.
Another way to state this is $V\otimes W\ket\psi=\ket{V\psi W^\top}$.
Note that in the case of indistinguishable subsystems we have $V=W$.
This suggests a connection to certain matrix diagonalizations, namely:
\begin{itemize}
\item The \emph{complex singular value decomposition} in the distinguishable subsystems case (often referred to as the Schmidt decomposition in the context of quantum mechanics). It states that for any complex matrix $\psi\in\C^{d_1,d_2}$ there exist unitary matrices $V\in\U(d_1)$ and $W\in\U(d_2)$ such that $V\psi W^*$ is real and diagonal. 
The correspondence is established by
$$
\ket\phi = V\otimes\overline W\ket\psi \iff \phi=V\psi W^*.
$$
\item The \emph{Autonne--Takagi factorization} in the bosonic case. It states that for any complex symmetric matrix $\psi\in\mf{sym}(\C,d)$, there is a unitary $V\in\U(d)$ such that $V\psi V^\top$ is real and diagonal.
The correspondence is then given by
$$
\ket\phi = V\otimes V\ket\psi \iff \phi=V\psi V^\top.
$$
\item The \emph{Hua factorization} in the fermionic case. It states that for any complex skew-symmetric matrix $\psi\in\mf{asym}(\C,d)$, there is a unitary $V\in\U(d)$ such that $V\psi V^\top$ is real and \emph{quasi-diagonal} in the following sense: if $d$ is even, then $V\psi V^\top$ is block diagonal with blocks of size $2\times2$, if $d$ is odd, there is an additional block of size $1\times1$ in the lower right corner. Note that the quasi-diagonal matrix is still skew-symmetric, and so the diagonal is zero. The correspondence to local unitary state transformations is as in the bosonic case.
\end{itemize}
Note that the Autonne--Takagi factorization and the Hua factorization are special cases of singular value decompositions, and hence the resulting (quasi\=/)diagonal matrix will have the singular values on its (quasi\=/)diagonal.
In the first case we denote by $\Sigma\subset\C^{d_1}\otimes\C^{d_2}$ the subspace of ``real diagonal states'' corresponding to the set of real diagonal matrices $\mf{diag}(d_1,d_2,\R)$ under the isomorphism~\eqref{eq:matrix-iso}.
Clearly $\Sigma$ has dimension $\dmin:=\min(d_1,d_2)$.
In the second case we write $\Sigma\subset\mathrm{Sym}^2(\C^d)$ for the $d$-dimensional subspace corresponding to the real diagonal matrices $\mf{diag}(d,\R)$.
Similarly, in the third case we write $\Xi\subset\bigwedge^2(\C^d)$ for the $\floor{d/2}$-dimensional subspace of states corresponding to the real (skew-symmetric) quasi-diagonal matrices $\mf{qdiag}(d,\R)$. 
We will use the following maps to send the singular values to their corresponding (quasi\=/)diagonal state: 
\begin{alignat*}{4}
\diag&:\R^{\dmin}\to\Sigma,   \quad & \,(\sigma_i)_{i=1}^{\dmin}&\mapsto \textstyle\sum_{i=1}^{\dmin}\sigma_i\ket{i}\otimes\ket{i}, \\
\qdiag&:\R^{\floor{d/2}}\to\Xi, \quad & \,(\xi_i)_{i=1}^{\floor{d/2}}&\mapsto \tfrac1{\sqrt2}\textstyle\sum_{i=1}^{\floor{d/2}}\xi_i(\ket{2i-1}\otimes\ket{2i}-\ket{2i}\otimes\ket{2i-1}).
\end{alignat*}
A convenient shorthand notation is $\ket\sigma = \diag(\sigma)$ resp.\ $\ket\xi = \qdiag(\xi)$.
We will always use the standard Euclidean inner product on $\R^n$, and on $\C^n$ we will use the real part $\Re(\braket{\cdot|\cdot})$ of the standard inner product.
Then, due to the inclusion of the factor $1/{\sqrt2}$ it holds that the maps above are $\R$-linear isometric isomorphisms.
Furthermore we denote\footnote{The symbol $\plusdown$ is a combination of $\down$ and $+$.} by $\Sigma_\plusdown\subset\Sigma$ the cone of states $\diag(\sigma)$ where the diagonal elements $(\sigma_i)_{i=1}^\dmin$ are non-negative and arranged in non-increasing order, and analogously we write $\Xi_\plusdown\subset\Xi$ for the quasi-diagonal states $\qdiag(\xi)$ where the $(\xi_i)_{i=1}^{\floor{d/2}}$ are non-negative and arranged in non-increasing order. 
The set $\Sigma_\plusdown$ resp.\ $\Xi_\plusdown$ is called the \emph{Weyl chamber}.

Conversely to the diagonal embeddings we also define the following orthogonal projections: 
\begin{align*}
\Pi_\Sigma &: \begin{cases}
\C^{d_1}\otimes\C^{d_2} \to \R^{\dmin},\,\ket\psi \mapsto (\braket{ii|\psi})_{i=1}^{\dmin} & \text{ for distinguishable subsystems} \\
\mathrm{Sym}^2(\C^d) \to \R^d,\,\ket\psi \mapsto (\braket{ii|\psi})_{i=1}^d  & \text{ for bosonic subsystems} \end{cases} \\
\Pi_\Xi &: \textstyle\bigwedge^2(\C^d) \to \R^{\floor{d/2}},  \, \ket\psi \mapsto \sqrt2 (\bra{2i}_2\bra{2i-1}_1\ket\psi)_{i=1}^{\floor{d/2}}\,.
\end{align*}
More precisely these are the orthogonal projections on $\Sigma$ and $\Xi$ followed by $\diag^{-1}$ and $\qdiag^{-1}$ respectively.

\begin{remark} \label{rmk:sphere}
The normalization of the quantum state entails a normalization of the corresponding singular values.
More precisely, the norm of the quantum state $\ket\psi$ equals the Frobenius norm of the matrix $\psi$, and hence $\ket\psi$ has unit norm if and only if the singular values $\sigma=(\sigma_i)_{i=1}^\dmin$ of $\psi$ satisfy $\sum_{i=1}^\dmin \sigma_i^2=1$.
Hence the singular values define a point on the unit sphere of dimension $\dmin-1$ in the indistinguishable and $d-1$ in the bosonic case.
In the fermionic case there are again $d$ singular values and they lie on the unit sphere of dimension $d-1$. 
However there is an additional restriction as the singular values come in pairs of opposite values. 
Taking only one singular value of each pair and multiplying it by $\sqrt2$ (and ignoring the $0$ singular value in the odd dimensional case)
we find that the resulting vector lies on the unit sphere of dimension $\floor{d/2}-1$. 
In all cases will call this the \emph{Schmidt sphere}, denoted $S^{\dmin-1}$ in the distinguishable case, $S^{d-1}$ in the bosonic case, and $S^{\floor{d/2}-1}$ in the fermionic case.
The maps $\diag$ and $\qdiag$ then yield isometric embeddings of the Schmidt sphere into $\Sigma$ and $\Xi$ respectively.
The Schmidt sphere will be the state space of our reduced control system, which we will define in the following section.
The Weyl chambers $\Sigma_\plusdown$ and $\Xi_\plusdown$ then yield corresponding Weyl chambers in the Schmidt sphere $S^{\dmin-1}_\plusdown$, $S^{d-1}_\plusdown$, and $S^{\floor{d/2}-1}_\plusdown$.
\end{remark}

\begin{remark}
Often one considers the Schmidt values, which are the squares of the singular values, within the standard simplex, which is then called the Schmidt simplex~\cite[Sec.~16.4]{Bengtsson17}. This is easier to visualize, but in our case would lead to unnatural dynamics, which is why we will remain on the sphere. Note also that if one of the systems is a qubit, then the Schmidt sphere is a circle parametrized by the Schmidt angle~\cite[p.~440]{Bengtsson17}.
\end{remark}

\subsection{Reduced Control System} \label{sec:reduced-system}

Due to Assumptions~\ref{it:fast-control} and~\ref{it:full-control} (resp.~\ref{it:full-control-2}), we can move arbitrarily quickly within the local unitary orbits of the system if we ignore the drift term~\cite[Prop.~2.7]{Elliott09}. 
With the drift this is still approximately true.
In the previous section we have shown that using local unitary transformations, we can always obtain a state of (quasi\=/)diagonal form which is completely determined by the singular values of the state. 
In particular, within the bilinear control systems~\eqref{eq:bilinear} and~\eqref{eq:bilinear2} two states are effectively equivalent if and only if they have the same singular values (up to order and sign). 
This strongly suggests that there should exist a ``reduced'' control system, defined on the singular values --- or rather the Schmidt sphere (cf.\ Remark~\ref{rmk:sphere}). 
This is indeed the case. 
The reduced control system is defined in greater generality in~\cite[Sec.~2.1]{Reduced23} using symmetric Lie algebras, which unify many well-known matrix diagonalizations, such as the ones encountered in the previous section. 
No knowledge of symmetric Lie algebras is presupposed here, but the connections are expounded in Appendix~\ref{app:sym-lie-alg}.

Let us briefly motivate the definition of the reduced control system.
Let $\ket\psi$ be a solution to the full control system~\eqref{eq:bilinear}.
Assume that the corresponding matrix $\psi$ can be diagonalized in a differentiable way as $\psi(t)=V(t)\tilde\sigma(t)W^\top(t)$, and that it is regular\footnote{We say that $\psi$ is regular if its singular values are distinct and non-zero.}.
Here $\tilde\sigma(t)$ is the diagonal matrix with diagonal elements $\sigma(t)$. 
Then by differentiating (cf.~\cite[Lem.~2.3]{Diag22}) we obtain that $\dot\sigma=-H_{V\otimes W}\sigma$ (and analogously $\dot\sigma=-H_{V\otimes V}^s\sigma$ or $\dot\xi=-H_{V\otimes V}^a\xi$ in the bosonic and fermionic cases) where
\begin{align*}
-H_{V\otimes W}   &:= -\Pi_\Sigma\circ (V\otimes W)^* \,\iu H_0 (V\otimes W)\circ\diag, \\
-H_{V\otimes V}^s &:= -\Pi_\Sigma\circ (V\otimes V)^* \,\iu H_0 (V\otimes V)\circ\diag, \\
-H_{V\otimes V}^a &:= -\Pi_\Xi\circ (V\otimes V)^* \,\iu H_0 (V\otimes V)\circ\qdiag.
\end{align*}
We call $H_{V\otimes W}$, $H_{V\otimes V}^s$ and $H_{V\otimes V}^a$ the \emph{induced vector fields}.
The collection of induced vector fields is denoted
$$
\mf H := \{-H_U:U\in\Uloc(d_1,d_2)\}, \,
\mf H^s := \{-H_U^s:U\in\Uloc^s(d)\}, \,
\mf H^a := \{-H_U^a:U\in\Uloc^s(d)\}.
$$
Note that these are linear vector fields on $\R^{\dmin}$, $\R^d$ and $\R^{\floor{d/2}}$ respectively, and hence they can be represented as matrices in the respective standard basis.
We will later see that these are indeed skew-symmetric matrices and that the corresponding dynamics preserve the Schmidt sphere.
The following proposition gives the explicit expressions.

\begin{proposition} \label{prop:reduced-hams}
Let $H_0\in\iu\mf u(d_1d_2)$ denote an arbitrary coupling Hamiltonian and let $E_k\in\iu\mf u(d_1)$ and $F_k\in\iu\mf u(d_2)$ for $k=1,\ldots,r$ be given such that $H_0=\sum_{k=1}^r E_k\otimes F_k$. 
Then the induced vector field $H_{V\otimes W}$ on $\R^{\dmin}$ 
takes the form%
\footnote{Here $\circ$ denotes the Hadamard (elementwise) product of two matrices. If the (square) matrices are of different size the resulting matrix will have the size of the smaller one. Similarly $\Im$ denotes the elementwise imaginary part.}
$$
-H_{V\otimes W}=\sum_{k=1}^r \Im(V^*E_kV\circ W^*F_kW).
$$
Now assume additionally that $d:=d_1=d_2$ and that $H_0\in\iu\mf{u}^s(d^2)$ is symmetric.
For bosonic systems, on $\R^d$, 
we obtain
$$
-H^s_{V\otimes V}=\sum_{k=1}^r \Im(V^*E_kV \circ V^*F_kV).
$$
For fermionic systems, on $\R^{\floor{d/2}}$, we obtain
$$
-(H^a_{V\otimes V})_{ij} =\sum_{k=1}^r \Im( (V^*E_kV)_{2i-1,2j-1} (V^*F_kV)_{2i,2j} - (V^*E_kV)_{2i-1,2j} (V^*F_kV)_{2i,2j-1}).
$$
\end{proposition}

\begin{proof}
To simplify notation we first consider a product Hamiltonian $H_0=E\otimes F$. 
Let $\sigma\in\R^{\dmin}$ be given.
Let $i,j\in\{1,\ldots,\dmin\}$ and compute
\begin{align*}
-(H_{V\otimes W})_{ij}
=
(-\Pi_\Sigma (\iu V^*EV e_je_j^\top W^\top F^\top (W^*)^\top))_i
=
\Im( (V^*EV)_{ij} (W^*FW)_{ij} ).
\end{align*}
The case of general $H_0$ for distinguishable subsystems as well as the bosonic case follow by linearity.

Now let us consider the fermionic case. 
Again for simplicity we consider a product Hamiltonian $H=E\otimes F$ since the general result follows by linearity. 
Then $E$ and $F$ can be seen as $n\times n$ block matrices with blocks of size $2\times 2$ denoted $E_{(ij)}$ and $F_{(ij)}$ (and, in the odd-dimensional case, an additional row and column).
Moreover let $J=\big[\begin{smallmatrix}0&1\\-1&0\end{smallmatrix}\big]$ denote the standard symplectic form.
Then we can compute in a similar fashion
\begin{align*}
-(H_{V\otimes V}^a)_{ij}
&=
(-\Pi_\Xi (\iu (V^*EV) \qdiag(e_j) (V^* F^\top V)^\top))_i
\\&=
\Im(((V^*EV)_{(ij)} J ((V^*FV)_{(ij)})^\top)_{12})\,.
\end{align*}
This concludes the proof.
\end{proof}

\begin{remark}
The fermionic case can be interpreted as follows.
We consider the even dimensional case for simplicity. 
First define the matrices $G_k$ obtained from $E_k$ and $F_k$ by choosing all the odd indexed rows from $E_k$ and all the even indexed rows from $F_k$. 
More precisely, $(G_k)_{ij} = (E_k)_{ij}$ if $i$ is odd and $(G_k)_{ij} = (F_k)_{ij}$ if $i$ is even. 
Then we divide the $G_k$ into blocks of size $2\times2$ and we define a new matrix of half the size by replacing each block by the imaginary part of its determinant. 
\end{remark}

We see immediately that if the drift Hamiltonian is local, then the induced vector fields vanish:

\begin{corollary} \label{coro:local-drift}
If the drift Hamiltonian $H_0$ is local, meaning that $H_0\in\iu\uloc(d_1,d_2)$, then $H_{V\otimes W}=0$ for every $V\otimes W\in\Uloc(d_1,d_2)$. The same is true for $H_{V\otimes V}^s$ and $H_{V\otimes V}^a$ whenever $H_0\in\iu\uloc^s(d)$ and $V\otimes V\in\Uloc^s(d)$.
\end{corollary}

Finally we can define the reduced control systems:

\begin{definition}[Reduced control systems]
Let $I\subset\R$ be an interval of the form $[0,T]$ with $T\geq0$ or $[0,\infty)$. 
We define three reduced control systems:

For distinguishable subsystems we set:
\begin{equation} \label{eq:reduced} \tag{\sf R}
\dot\sigma(t) = -H_{U(t)}\sigma(t), \quad \sigma(0)=\sigma_0\in S^{\dmin-1}
\end{equation}
A solution is an absolutely continuous path $\sigma:I\to S^{\dmin-1}$ satisfying~\eqref{eq:reduced} almost everywhere for some measurable control function $U:I\to\Uloc(d_1,d_2)$.

In the bosonic case we set:
\begin{equation} \label{eq:reduced-sym} \tag{${\sf R}^s$}
\dot\sigma(t) = -H^s_{U(t)}\sigma(t), \quad \sigma(0)=\sigma_0\in S^{d-1}
\end{equation}
A solution is an absolutely continuous path $\sigma:I\to S^{d-1}$ satisfying~\eqref{eq:reduced-sym} almost everywhere for some measurable control function $U:I\to\Uloc^s(d)$.

In the fermionic case we set:
\begin{equation} \label{eq:reduced-asym} \tag{${\sf R}^a$}
\dot\xi(t) = -H^a_{U(t)}\xi(t), \quad \xi(0)=\xi_0\in S^{\floor{d/2}-1}
\end{equation}
A solution is an absolutely continuous path $\xi:I\to S^{\floor{d/2}-1}$ satisfying~\eqref{eq:reduced-asym} almost everywhere for some measurable control function $U:I\to\Uloc^s(d)$.
\end{definition}
As mentioned previously, and shown in Lemma~\ref{lemma:skew-sym} below, the induced vector fields preserve the Schmidt sphere, and hence we may define the reduced control systems directly on the Schmidt sphere.

\begin{remark} \label{rmk:other-defs}
There are several slightly different ways of defining the reduced control system which are given in~\cite[Sec.~2.1]{Reduced23}. 
The most intuitive definition, given above, is to consider the control system $\dot{\sigma}(t)=-H_{U(t)}(\sigma(t))$ where the control function $U:[0,T]\to\Uloc(d_1,d_2)$ is measurable and the solution $\sigma:[0,T]\to S^{\dmin-1}$ is absolutely continuous. 
A more geometric definition uses the differential inclusion $\dot{\sigma}(t)\in\derv(\sigma(t))$, where $\derv(\sigma)$ denotes the set of \emph{achievable derivatives} at $\sigma$ defined by $\derv(\sigma)=\{-H_U\sigma:U\in\Uloc(d_1,d_2)\}$.
The differential inclusion is exactly equivalent to our definition by Filippov's Theorem, cf.~\cite[Thm.~2.3]{Smirnov02}.
Often it is convenient to consider a ``relaxed'' version of the differential inclusion where also convex combinations of achievable derivatives are allowed: $\dot{\sigma}(t)\in\conv(\derv(\sigma(t)))$. This slightly enlarges the set of solutions, but every solution to the relaxed system can still be approximated uniformly on compact time intervals by solutions to our system, cf.~\cite[Ch.~2.4, Thm.~2]{Aubin84}.
Analogous remarks also hold for the symmetric cases~\eqref{eq:reduced-sym} and~\eqref{eq:reduced-asym}.
\end{remark}

The main result of~\cite{Reduced23} is the equivalence of the full bilinear control system~\eqref{eq:bilinear} resp.~\eqref{eq:bilinear2} and the reduced control system~\eqref{eq:reduced}, resp.~\eqref{eq:reduced-sym} or~\eqref{eq:reduced-asym}, proven in Theorems 3.8 and 3.14 therein.
In our case this specializes to the following result.

First we need to define the quotient maps $\mathrm{sing}^\plusdown:\C^{d_1}\otimes\C^{d_2}\to\R^{\dmin}$ and $\mathrm{qsing}^\plusdown:\bigwedge^2(\C^d)\to\R^{\floor{d/2}}$.  
Given a (possibly not normalized) vector $\ket\psi\in\C^{d_1}\otimes\C^{d_2}$ (or $\mathrm{Sym}^2(\C^d)$), the map $\mathrm{sing}^\plusdown$ yields the singular values of the corresponding matrix $\psi\in\C^{d_1,d_2}$, chosen non-negative and arranged in non-increasing order.
Similarly, for $\ket\xi\in\bigwedge^2(\C^d)$, the map $\mathrm{qsing}^\plusdown$ yields the singular values of the skew-symmetric matrix $\xi$, except that we keep only one singular value of each pair and multiply it by $\sqrt2$ to keep the normalization. 
Note that when restricting the domain of $\mathrm{sing}^\plusdown$ and $\mathrm{qsing}^\plusdown$ to (normalized) quantum states, the image will lie in the respective Schmidt sphere, and even in the Weyl chamber $S_\plusdown^{\dmin-1}$ resp.\ $S_\plusdown^{d-1}$ and $S_\plusdown^{\floor{d/2}-1}$.

Recall that here and throughout the paper we use Assumptions~\ref{it:fast-control} and~\ref{it:full-control} (resp.~\ref{it:full-control-2}), unless stated otherwise.  

\begin{theorem}[Equivalence Theorem] \label{thm:equivalence}
Let $\ket{\psi(t)}$ be a solution on $[0,T]$ to the bilinear control system~\eqref{eq:bilinear}, and let $\sigma^\plusdown:[0,T]\to S_\plusdown^{\dmin-1}$ be defined by $\sigma^\plusdown=\mathrm{sing}^\plusdown(\ket\psi)$. 
Then $\sigma^\plusdown$ is a solution to the reduced control system~\eqref{eq:reduced}.

Conversely, let $\sigma:[0,T]\to S^{\dmin-1}$ be a solution to the reduced control system~\eqref{eq:reduced} with control function $U:[0,T]\to\Uloc(d_1,d_2)$ and let $\ket\sigma=\diag(\sigma)$ denote the corresponding state. 
Then $U(t)\ket{\sigma(t)}$ can be approximated by solutions to the full control system~\eqref{eq:bilinear} arbitrarily well.
More precisely, for every $\varepsilon>0$ there exists a solution $\ket{\psi_\varepsilon(t)}$ to~\eqref{eq:bilinear} such that $\|U\ket{\sigma}-\ket{\psi_\varepsilon}\|_\infty\leq\varepsilon$, where $\|\cdot\|_\infty$ denotes the supremum norm.

The analogous results, \emph{mutatis mutandis}\footnote{Most results in this paper hold in all three cases with only minimal differences in notation, which are summarized in Table~\ref{tab:notation}.}, also hold in the bosonic and fermionic cases, where the full control system is~\eqref{eq:bilinear2} and the reduced control systems are~\eqref{eq:reduced-sym} and~\eqref{eq:reduced-asym} respectively.
\end{theorem}

\begin{proof}
The proof is mostly a technicality, as we simply have to show that the full control systems~\eqref{eq:bilinear} and~\eqref{eq:bilinear2} and their respective reduced versions~\eqref{eq:reduced}, \eqref{eq:reduced-sym} and~\eqref{eq:reduced-asym} can be interpreted as control systems on certain symmetric Lie algebras.
We focus on the case of distinguishable subsystems.
The corresponding symmetric Lie algebra is that of type AIII.
The isomorphisms $\imath^d$ and $\jmath^d$ defined in Appendix~\ref{app:AIII} translate the quantum setting into the Lie algebra setting.
The results of the appendix then show that all of the conditions of Proposition~\ref{prop:equivalence} are satisfied and that the reduced control system~\eqref{eq:reduced} indeed corresponds to the reduced control system in the Lie algebraic setting.
Taken together this proves the equivalence in the distinguishable case.
The other cases are entirely analogous.
\end{proof}

The Equivalence Theorem~\ref{thm:equivalence} shows that the full bilinear control system~\eqref{eq:bilinear} resp.~\eqref{eq:bilinear2} and the reduced control system~\eqref{eq:reduced}, resp.~\eqref{eq:reduced-sym} or~\eqref{eq:reduced-asym}, contain essentially the same information.
Hence for every control theoretic notion, such as controllability and stabilizability, there is a specialized equivalence result, see~\cite[Sec.~4]{Reduced23} for an overview.
As a first consequence we obtain:

\begin{lemma} \label{lemma:skew-sym}
The induced vector fields are skew-symmetric matrices:
$$
\mf H\subset\mf{so}(\dmin,\R), \quad
\mf H^s\subset\mf{so}(d,\R), \quad
\mf H^a\subset\mf{so}(\floor{d/2},\R).
$$
In particular the Schmidt sphere in invariant.
\end{lemma}

\begin{proof}
Due to the Equivalence Theorem~\ref{thm:equivalence} this follows from~\cite[Prop.~4.8]{Reduced23}.
Alternatively this can also be verified by direct computation using the expressions obtained in Proposition~\ref{prop:reduced-hams}.
\end{proof}

Let us also recall the equivalence of reachable sets here, which is arguably the most useful consequence. 
First we give the definitions of reachable sets in the reduced control system~\eqref{eq:reduced}. The definitions for other control systems are entirely analogous. The \emph{reachable set of $\sigma_0$ at time $T$} is defined as 
$$
\reach_{\ref{eq:reduced}}(\sigma_0,T) = \{\sigma(T) : \sigma:[0,T]\to S^{\dmin-1}\text{ solves }\eqref{eq:reduced}, \sigma(0)=\sigma_0 \}
$$
for any $T\geq0$. By $\reach_{\ref{eq:reduced}}(\sigma_0):=\bigcup_{T\geq0}\reach_{\ref{eq:reduced}}(\sigma_0,T)$ we denote the \emph{all time reachable set of $\sigma_0$}, and by $\reach_{\ref{eq:reduced}}(\sigma_0,[0,T]):=\bigcup_{t\in[0,T]}\reach_{\ref{eq:reduced}}(\sigma_0,t)$ we denote the \emph{reachable set of $\sigma_0$ up to time $T$}.

The following result is an immediate consequence of the Equivalence Theorem~\ref{thm:equivalence} and~\cite[Prop.~4.3]{Reduced23}.

\begin{proposition} \label{prop:reach}
Let $T>0$ be given and assume that $\ket{\psi_0}\in\C^{d_1}\otimes\C^{d_2}$ and $\sigma_0\in S^{\dmin-1}$ satisfy\footnote{Here $\sigma_0^\plusdown$ denotes the element of $S^{\dmin-1}_\plusdown$ whose elements are the absolute values of the elements of $\sigma_0$ arranged in non-increasing order.} $\sigma_0^\plusdown=\mathrm{sing}^\plusdown(\psi_0)$.
Then is holds that 
$$
\reach_{\ref{eq:bilinear}}(\ket{\psi_0},T) 
\subseteq 
\{ U\ket{\sigma} : \ket\sigma\in\reach_{\ref{eq:reduced}}(\sigma_0,T), \,U\in\Uloc(d_1,d_2)\}
\subseteq 
\overline{\reach_{\ref{eq:bilinear}}(\ket{\psi_0},T)}\,.
$$
In particular, the closures coincide:
$$
\overline{\reach_{\ref{eq:bilinear}}(\ket{\psi_0},T)} = 
\overline{\{ U\ket{\sigma} : \ket\sigma\in\reach_{\ref{eq:reduced}}(\sigma_0,T), \,U\in\Uloc(d_1,d_2)\}}\,.
$$
The analogous result, \emph{mutatis mutandis}, holds also for the bosonic and the fermionic cases.
\end{proposition}

The Equivalence Theorem~\ref{thm:equivalence} guarantees the existence of an approximate lift, but its proof also provides a way to find corresponding control functions. 
Under some additional assumptions we can give an explicit formula for the controls of an exact lift, see~\cite[Prop.~3.10]{Reduced23}.
In particular this requires the solution to be smooth and regular, and the controls of the bilinear system~\eqref{eq:bilinear} resp.~\eqref{eq:bilinear2} to linearly span the corresponding Lie algebra.
Before stating the result we define some notation.

The group of local unitary operations $\Uloc(d_1,d_2)$ and its symmetric counterpart $\Uloc^s(d)$ act on the state spaces $\C^{d_1}\otimes\C^{d_2}$, $\mathrm{Sym}^2(\C^d)$ and $\bigwedge^2(\C^d)$ respectively.
The corresponding infinitesimal action of the Lie algebras $\uloc(d_1,d_2)$ and $\uloc^s(d)$ can be determined as in Section~\ref{sec:matrix-decs} using the formula $(\iu E\otimes\id+\id\otimes\iu F)\ket\psi = \ket{\iu(E\psi+\psi\overline F)}$.
In the special case where the state is regular and diagonal this function (more precisely its negative) gets a special name:
\begin{alignat*}{3}
\ad_\sigma^d &: \uloc(d_1,d_2) \to \C^{d_1}\otimes\C^{d_2}, &\,\,&\iu E\otimes\id+\id\otimes\iu F \mapsto -\ket{\iu(E\tilde\sigma+\tilde\sigma\overline F)} \\
\ad_\sigma^s &: \uloc^s(d) \to \mathrm{Sym}^2(\C^d),         &&\iu E\otimes\id+\id\otimes\iu E \mapsto -\ket{\iu(E\tilde\sigma+\tilde\sigma\overline E)} \\
\ad_\xi^a      &: \uloc^s(d) \to \textstyle\bigwedge^2(\C^d), &&\iu E\otimes\id+\id\otimes\iu E \mapsto -\ket{\iu(E\tilde\xi+\tilde\xi\overline E)},
\end{alignat*}
where $\tilde\sigma\in\mf{diag}(d_1,d_2,\R)$ and $\tilde\xi\in\mf{qdiag}(d,\R)$ denote the (quasi-)diagonal matrices corresponding to $\sigma$ and $\xi$.  
Although these maps are not bijective, by restricting the domain to the orthocomplement of the kernel and the codomain to the image, inverse maps can be defined.
Indeed, this is nothing but the Moore--Penrose pseudoinverse.
Explicit expressions are given in Lemmas~\ref{lemma:ad-d-inv}, \ref{lemma:ad-s-inv} and~\ref{lemma:ad-a-inv}.

To use these inverse maps we have to understand the images of the maps $\ad_\sigma^d$, $\ad_\sigma^s$, and $\ad_\xi^a$.
It turns out that, for regular $\sigma$ resp.\ $\xi$, these images are exactly given by the orthocomplement of the diagonal subspaces $\Sigma,\Xi$.
We denote the orthogonal projection on $\C^{d_1}\otimes\C^{d_2}$ with kernel $\Sigma$ by $\Pi_\Sigma^\perp$, and use the same notation on $\mathrm{Sym}^2(\C^d)$. 
In the matrix picture, this map simple removes the real part of the diagonal elements of $\psi$.
Similarly, on $\bigwedge^2(\C^d)$, the orthogonal projection with kernel $\Xi$ is denoted $\Pi_\Xi^\perp$ and it removes the real part of the quasi-diagonal of $\psi$.

With these definitions~\cite[Prop.~3.10]{Reduced23} can be specialized as follows:

\begin{proposition} \label{prop:exact-lift}
Let $\sigma:[0,T]\to S^{\dmin-1}$ be a solution to the reduced control system~\eqref{eq:reduced} with control function $U:[0,T]\to\Uloc(d_1,d_2)$. 
Assume that $\sigma$ is regular and that $U$ is continuously differentiable. 
Let $\ket{\psi(t)}=U(t)\ket{\sigma(t)}$ and let $H:[0,T]\to\iu\mf\uloc(d_1,d_2)$ be given by
$$
-\iu H(t) = \dot U(t)U^{-1}(t) - U(t) ((\mad_{\sigma(t)})^{-1} \circ \Pi_\Sigma^\perp) \big(U(t)^*(\iu H_0)U(t)\ket{\sigma(t)}\big).
$$
Then $\ket\psi$ satisfies $\ket{\dot\psi}=-\iu(H_0+H)\ket\psi$.
The analogous result also holds \emph{mutatis mutandis} for bosonic and fermionic systems.
\end{proposition}
\noindent We call the second term in the definition of $H$ in Proposition~\ref{prop:exact-lift} the \emph{compensating Hamiltonian} since it compensated for the local Hamiltonian action induced by the drift term $H_0$.

If the control directions linearly span the entire Lie algebra $\uloc(d_1,d_2)$, then it is easy to find the control functions from Proposition~\ref{prop:exact-lift}.
More generally the problem of finding corresponding controls is studied under the term non-holonomic motion planning, see~\cite{Liu97}.

\subsection{Global Phases} \label{sec:global-phase}

Mathematically the quantum state $\ket\psi\in\C^{d_1}\otimes\C^{d_2}$ has a global phase which is physically undetectable and hence may be considered irrelevant.
We keep the phase throughout this paper for convenience, noting that the global phase is removed automatically in the reduced control system.
This section briefly discusses how Assumptions~\ref{it:full-control} and~\ref{it:full-control-2} can be slightly weakened by neglecting the global phase.
We denote by $\suloc(d_1,d_2)$ and $\suloc^s(d)$ the Lie algebras obtained from $\uloc(d_1,d_2)$ and $\uloc^s(d)$ by requiring the trace to vanish.
Consider the following two weakened controllability assumptions:

\begin{enumerate}[(I'),start=2]
\item \label{it:full-control-weak} The control Hamiltonians generate the local special unitary Lie algebra:
$$\generate{\iu E_i\otimes\id,\id\otimes\iu F_j : \,i=1,\ldots,m_1,j=1,\ldots,m_2}{\mathsf{Lie}}=\suloc(d_1,d_2).$$
\item \label{it:full-control-2-weak} The control Hamiltonians generate the symmetric local special unitary Lie algebra:
$$\generate{\iu E_i\otimes\id+\id\otimes\iu E_i : \,i=1,\ldots,m}{\mathsf{Lie}}=\suloc^s(d).$$
\end{enumerate}
In both cases adding the (symmetric) local control Hamiltonian $\id\otimes\id$ is sufficient to obtain the stronger Assumptions~\ref{it:full-control} and~\ref{it:full-control-2} respectively.
Since this Hamiltonian commutes with everything, the only effect of adding or removing the corresponding term from the control system is a change in the global phase.
More concretely, if our control system is~\eqref{eq:bilinear} (resp.~\eqref{eq:bilinear2}) but only satisfies Assumptions~\ref{it:fast-control} and~\ref{it:full-control-weak} (resp.~\ref{it:full-control-2-weak}), then we can add the control Hamiltonian $\id\otimes\id$ so that it satisfies Assumption~\ref{it:full-control} (resp.~\ref{it:full-control-2}). 
Now we can compute any solution in this extended system and we obtain a corresponding solution in the actual system by setting the control function of $\id\otimes\id$ to zero.
The resulting solution will, at all times, be equal to the solution of the extended system up to a global phase.

\section{Some Applications} \label{sec:applications}

In the previous section we defined the reduced control systems~\eqref{eq:reduced},~\eqref{eq:reduced-sym}, and~\eqref{eq:reduced-asym} and  proved that they are equivalent to the corresponding full control system~\eqref{eq:bilinear} in the distinguishable case and~\eqref{eq:bilinear2} in the bosonic and fermionic cases, cf. Theorem~\ref{thm:equivalence}.
In this section we will use the reduced control system to deduce some consequences on controllability, stabilizability and speed limits.
Using the Equivalence Theorem~\ref{thm:equivalence} these results can then be translated to the full control system.

\subsection{Controllability and Stabilizability} \label{sec:ctrl-stab}

In this section we show that the reduced control system is always controllable and stabilizable. 
As a consequence, the full control system is also controllable and all states can be stabilized in a certain sense.
For notational simplicity we focus on the case of distinguishable subsystems, noting that the bosonic and fermionic cases are entirely analogous.

First we lift the reduced control system~\eqref{eq:reduced} to the Lie group $\SO(\dmin)$. 
Unless otherwise noted it is assumed that the initial state is $R(0)=\id$.
The operator lift can be defined by
\begin{equation} \label{eq:lift} \tag{\sf L}
\dot R(t)=-H_{U(t)}R(t),
\end{equation}
and analogously with $H_U^s$ and $H_U^a$ in the bosonic and fermionic cases.
A solution is absolutely continuous and satisfies~\eqref{eq:lift} almost everywhere for some measurable $U$.

\begin{remark}
Note that the operator lift~\eqref{eq:lift} of the reduced control system is a useful but somewhat artificial construction.
Indeed, even though~\eqref{eq:bilinear} and~\eqref{eq:reduced} are equivalent, the operator lift~\eqref{eq:lift} is \emph{not} equivalent to the operator lift of~\eqref{eq:bilinear}. 
\end{remark}

The reduced control system is \emph{controllable on $S^{\dmin-1}$} if for every two states $\sigma_1,\sigma_2\in S^{\dmin-1}$ it holds that $\sigma_2\in\reach_{\ref{eq:reduced}}(\sigma_1)$ and it is \emph{controllable on $S^{\dmin-1}$ in time $T$} if for every two states $\sigma_1,\sigma_2\in S^{\dmin-1}$ we have that $\sigma_2\in\reach_{\ref{eq:reduced}}(\sigma_1,[0,T])$. \emph{Approximate controllability} is defined in the same way except that one considers the closures of the respective reachable sets.
The analogous definitions also holds for all other control systems.

To understand the properties of the operator lift, we study the set of generators $\vecs\subset\mf{so}(\dmin,\R)$.
A key property of $\vecs$ is that it is invariant under conjugation by the Weyl group $\mb W=S_n\wr\mathbb Z_2$, see~\cite[Lem.~A.2]{Reduced23} and Appendix~\ref{app:sym-lie-alg}. 
This fact allows us to prove the following result.

\begin{proposition} \label{prop:lift-controllable}
The Weyl group $\mb W$ acts irreducibly on $\mf{so}(\dmin,\R)$.
In particular, if the coupling Hamiltonian $H_0$ is not local, then the operator lift~\eqref{eq:lift} is controllable.
The analogous result holds, \emph{mutatis mutandis}, in the bosonic and fermionic cases.
\end{proposition}

\begin{proof}
To show that the Weyl group acts irreducibly, we start with an arbitrary non-zero element $\Omega\in\mf{so}(\dmin,\R)$ and show that the subrepresentation generated by $\Omega$ is all of $\mf{so}(\dmin,\R)$.
If $\dmin=2$ this is trivially true since $\mf{so}(2,\R)$ is one-dimensional.
So assume that $\dmin\geq3$.
Consider the basis $\{e_{ij}=E_{ij}-E_{ji}: 1\leq i<j\leq\dmin\}$ and let $\Omega_{ij}$ be the coefficients of $\Omega$ in this basis.
Since $\Omega$ is non-zero, at least one of the coefficients is non-zero. 
Using a permutation in $\mb W$ we may assume that $\Omega_{12}\neq0$. 
Let $W_i\in\mb W$ be the diagonal matrix whose diagonal equals $1$ everywhere except in the $i$-th position, where it equals $-1$. 
Consider the matrix $\Omega'=\frac{\Omega+W_3\Omega W_3^\top}{2}$.
Then $\Omega'_{12}=\Omega_{12}$ and $\Omega'_{i3}=\Omega'_{3j}=0$.
Iterating this procedure with $W_4,\dots,W_{\dmin}$ we obtain a multiple of $e_{12}$, showing that $e_{12}$ lies in the subrepresentation generated by $\Omega$. 
From this, using the permutations in $\mb W$, all other basis elements $e_{ij}$ can be obtained.
This shows that the representation of $\mb W$ is irreducible.
Controllability of the operator lift then follows from~\cite[Prop.~4.19]{Reduced23}.
\end{proof}

This result can now be lifted to the full bilinear system using the equivalence of the systems. 

\begin{theorem} \label{thm:controllability} 
If the full control system~\eqref{eq:bilinear} is controllable in time $T$, then the reduced control system~\eqref{eq:reduced} is controllable in time $T$ on the Weyl chamber $S^{\dmin-1}_\plusdown$.
Conversely, if the reduced control system~\eqref{eq:reduced} is controllable in time $T$ on the Weyl chamber $S^{\dmin-1}_\plusdown$, then the full control system~\eqref{eq:bilinear} is controllable in time $T+\varepsilon$ for all $\varepsilon>0$.
Moreover there exists a finite time $T$ such that both systems are controllable in time $T$.
The analogous result holds, \emph{mutatis mutandis}, in the bosonic and fermionic cases.
\end{theorem}

\begin{proof}
Clearly if~\eqref{eq:bilinear} is controllable in time $T$, so is~\eqref{eq:reduced} on the Weyl chamber by Proposition~\ref{prop:reach}. 
The same result shows that if~\eqref{eq:reduced} is controllable in time $T$ on the Weyl chamber, then~\eqref{eq:bilinear} is approximately controllable in time $T$.
Proposition~\ref{prop:lift-controllable} shows in particular that~\eqref{eq:reduced} is directly accessible at every point (see~\cite[Sec.~4.6]{Reduced23} for the definitions related to accessibility) and hence by~\cite[Prop.~4.15]{Reduced23} the bilinear system~\eqref{eq:bilinear} is accessible (on the set of normalized states) at every regular state. 
Then~\cite[Ch.~3 Thm.~2]{Jurdjevic97} implies that~\eqref{eq:bilinear} is controllable in time $T$ for regular initial states. 
Since regular states are dense and since we can leave the set of non-regular states in an arbitrarily short amount of time $\varepsilon$, the full control system~\eqref{eq:bilinear} is controllable in time $T+\varepsilon$.
That~\eqref{eq:reduced} is controllable in finite time follows from Proposition~\ref{prop:lift-controllable}.
\end{proof}

Now we turn to stabilizability.
In accordance with the definitions given in~\cite[Sec.~4.3]{Reduced23} we say that a state $\sigma\in S^{\dmin-1}$ is \emph{stabilizable} for~\eqref{eq:reduced} if $0\in\conv(\vecs\sigma)$.
A direct consequence of Proposition~\ref{prop:lift-controllable} is that every state is stabilizable.

\begin{corollary}
Every state is stabilizable for the reduced control systems. 
\end{corollary}

\begin{proof} 
If $H_0$ is local the statement trivially holds. 
Otherwise choose some non-zero $H_U$.
Consider the uniform combination $\hat H_U=\frac{1}{|\mb W|}\sum_{w\in\mb W} w H_{U} w^{-1}\in\conv(\vecs)$ which is clearly $\mb W$-invariant.
If $\dmin=2$, it is clear that $\hat H_U=0$.
In higher dimensions $\mb W$-invariance and irreducibility of the action of $\mb W$ (Proposition~\ref{prop:lift-controllable}) again show that $\hat H_U=0$. 
The proof for the bosonic and fermionic cases is the same.
\end{proof}

If a state $\sigma\in S^{\dmin-1}$ is stabilizable for~\eqref{eq:reduced}, then in the bilinear control system~\eqref{eq:bilinear}, one can stay close to the local unitary orbit $\Uloc(d_1,d_2)\ket\sigma$ for arbitrary amount of time, cf.~\cite[Prop.~4.7]{Reduced23}.

Recall that a point $\sigma\in\Sigma$ is \emph{strongly stabilizable} if there is $U\in\Uloc(d_1,d_2)$ such that $H_U\sigma=0$. 
Specializing~\cite[Prop.~4.5]{Reduced23} we obtain the following result.

\begin{proposition}
Let $\sigma\in S^{\dmin-1}$ and $U\in\Uloc(d_1,d_2)$ and set $\ket\psi=U\ket\sigma$. 
If there is some $H\in\iu\uloc(d_1,d_2)$ satisfying $(H_0+H)\ket{\psi}=0$, then $H_U\sigma=0$ and $\sigma$ is strongly stabilizable.

Conversely let $\sigma\in S^{\dmin-1}$ be strongly stabilizable and let $U\in\Uloc(d_1,d_2)$ be such that \mbox{$H_U\sigma=0$}. 
Moreover assume that $\sigma$ is regular. 
Then there is some $H\in\iu\uloc(d_1,d_2)$ such that $(H_0+H)\ket{\psi}=0$ where we again set $\ket\psi=U\ket\sigma$. In fact one can choose
$$
-\iu H(t) = - U(t) ((\ad^d_{\sigma(t)})^{-1} \circ \Pi_\Sigma^\perp) \big(U(t)^*(\iu H_0)U(t)\ket{\sigma(t)}\big),
$$
where $(\ad^d_{\sigma})^{-1}$ is given explicitly in Lemma~\ref{lemma:ad-d-inv}.
The analogous result holds, \emph{mutatis mutandis}, in the bosonic and fermionic cases.
\end{proposition}
\noindent Note that the assumption on regularity is necessary in general, cf.~\cite[Ex.~3.12]{Reduced23}. 

The local Hamiltonian $H$ in the previous result is called a \emph{compensating Hamiltonian}, and indeed this is a special case of Proposition~\ref{prop:exact-lift}. 
Note that the expressions in Lemmas~\ref{lemma:ad-d-inv}, \ref{lemma:ad-s-inv} and~\ref{lemma:ad-a-inv}, and hence the compensating Hamiltonian, blow up as $\sigma$ approaches a non-regular state.

The following result yields a simple special case in which strong stabilizability is easy to determine.

\begin{lemma}
Let $H_0=\sum_{i=1}^m E_i\otimes F_i$ and assume that all $E_i$ commute or that all $F_i$ commute.
Then there exists $U\in\Uloc(d_1,d_2)$ such that $H_U\equiv0$. 
In particular, in this case every state is strongly stabilizable.
The analogous result holds, \emph{mutatis mutandis}, in the bosonic and fermionic cases.
\end{lemma}



\subsection{Speed Limit and Control Time} \label{sec:speed-limit}


By \emph{speed limit} we simply mean an upper bound on the velocity that any solution to the given control system can achieve.
Note that the full control system~\eqref{eq:bilinear} (resp.~\eqref{eq:bilinear2}) does not have any such speed limit, since the controls may be unbounded, but, by construction, the reduced control system~\eqref{eq:reduced} (resp.~\eqref{eq:reduced-sym} and~\eqref{eq:reduced-asym}) always admits a (finite) speed limit, cf.~\cite[Prop.~4.1]{Reduced23}.

For any matrix $\Omega\in\R^{n\times n}$, we write $\|\Omega\|_{\infty}$ for the largest singular value of $\Omega$. 
This is exactly the operator norm with respect to the usual Euclidean norm, and hence it is clear that for $\Omega\in\mf{so}(n,\R)$, the norm $\|\Omega\|_{\infty}$ corresponds to the largest velocity that $\Omega$ achieves on the unit sphere. 
This immediately yields the following result:

\begin{lemma}
Let $\sigma:[0,T]\to S^{\dmin-1}$ 
be any solution to~\eqref{eq:reduced}. 
Then it holds that $\|\dot\sigma(t)\|\leq\max_{U}\|H_U\|_{\infty}$ almost everywhere.\footnote{The maximum exists and is achieved since the map $U\mapsto \|H_U\|_{\infty}$ is continuous on a compact domain.}
The analogous result holds, \emph{mutatis mutandis}, in the bosonic and fermionic cases.
\end{lemma}
\noindent Hence we need to find a good upper bound for $\|H_U\|_{\infty}$ over all $U\in\Uloc(d_1,d_2)$.

\begin{lemma}
Let $H_0=\sum_{k=1}^r E_k\otimes F_k$. Then
$$
\max_{U\in\Uloc(d_1,d_2)}\|H_U\|_{\infty} 
\leq 
\sqrt{\sum_{k=1}^r\|E_k\|^2_2\|F_k\|^2_2},
$$
where $\|A\|_2=\sqrt{\tr(A^*A)}$ denotes the Frobenius norm. The same bound holds a fortiori for $\max_{U\in\Uloc^s(d)}\|H_U^s\|_{\infty}$ and $\max_{U\in\Uloc^s(d)}\|H_U^a\|_{\infty}$.
\end{lemma}

\begin{proof}
The Frobenius norm $\|\cdot\|_2$ and the spectral norm $\|\cdot\|_\infty$ are related by $\|\cdot\|_\infty\leq\|\cdot\|_2$, 
see~\cite[Prob.~5.6.P23]{HJ1ed2}.
Using the Cauchy-Schwarz inequality we compute for any $U=V\otimes W\in\Uloc(d_1,d_2)$ that
\begin{align*}
\|H_U\|_{\infty}^2
\leq
\sum_{i,j=1}^\dmin|(H_U)_{ij}|^2
\leq 
\sum_{k=1}^r \sum_{i,j=1}^\dmin |(V^*E_kV)_{ij}|^2 |(W^*F_kW)_{ij}|^2
\leq 
\sum_{k=1}^r \|E_k\|^2_2 \|F_k\|^2_2.
\end{align*}
This concludes the proof in the case of distinguishable subsystems.
The bound continues to hold in the bosonic and fermionic cases since restricting the drift or the controls cannot lead to faster evolution of the singular values.
\end{proof}

To obtain an lower limit on the time needed to reach any target state from any initial state, we also need to know the largest distance between any pair of points.
This is the diameter of the space\footnote{Note that distances in the reduced state space are computed as the length of the shortest geodesic joining two points on the sphere. Hence, somewhat unintuitively, the diameter of a unit hypersphere is $\pi$, which is the distance between two antipodal points.}, and due to the Weyl group symmetry every state has an equivalent state in the Weyl chamber. 
Hence we are particularly interested in the diameter of the Weyl chamber, which is given in the following result.

\begin{lemma}
Consider the unit sphere $S^{n-1}$ embedded in $\R^n$ and let $\mb W=S_n\wr\mathbb Z_2$ be the Weyl group acting by coordinate reflections and permutations. Then the corresponding Weyl chamber has diameter $\arccos(\tfrac1{\sqrt n})\in [\tfrac\pi4,\tfrac\pi2)$.
\end{lemma}

\begin{proof} 
First recall that the shortest distance on the sphere between two points $x,y\in S^{n-1}$ is given by $\arccos(x^\top y)$.
The maximal distance in the Weyl chamber is achieved by two of its corners.
Then it is clear that these points are $x=(1,0,\ldots,0)$ and $y=(\tfrac{1}{\sqrt{n}},\ldots,\tfrac{1}{\sqrt{n}})$ for the standard Weyl chamber $S^{n-1}_\plusdown$.
The result follows immediately.
\end{proof}

The \emph{control time $T^\star$} of a control system is the shortest (infimum) time sufficient to reach any state from any other state.
Using the upper bound on the speed of a solution, the lower bound on the diameter of the Weyl chamber and Theorem~\ref{thm:controllability}, we can give a lower bound on the control time:

\begin{proposition} 
The control time $T^\star$ of the full control system~\eqref{eq:bilinear} (resp.~\eqref{eq:bilinear2}) is finite and satisfies
$$
T^\star 
\geq 
\frac{\pi/4}{\max_{U}\|H_U\|_{\infty}} 
\geq 
\frac{\pi}{4\sqrt{\sum_k\|E_k\|^2_2\|F_k\|^2_2}}\,.
$$
\end{proposition}



\section*{Acknowledgments}

I would like to thank Frederik vom Ende, Thomas Schulte-Herbrüggen and Gunther Dirr for valuable and constructive comments during the preparation of this manuscript.
The project was funded i.a.~by the Excellence Network of Bavaria ENB under the International PhD Programme of Excellence
\textit{Exploring Quantum Matter} (ExQM) and by the {\it Munich Quantum Valley} of the Bavarian State Government with funds from Hightech Agenda {\it Bayern Plus}.

\appendix

\section{Relation to Symmetric Lie Algebras} \label{app:sym-lie-alg}

In the main text we have shown that the local unitary actions on bipartite quantum states correspond to certain matrix diagonalizations, and we have stated that they themselves are related to certain symmetric Lie algebras. 
In this appendix we make these relations explicit and give all the relevant formulas.
For a compact overview of the relation of symmetric Lie algebras to matrix diagonalizations see~\cite[Tab.~2]{Diag22}.

Since we want to define a reduced control system on the singular values, a key question is how the singular values change in time. 
More precisely, given a differentiable path of matrices $\psi(t)$, what can we say about the derivative of the singular values? 
This question is made more complicated by the fact that the order and signs of the singular values are not unique (and if they are chosen in a unique way, they are not guaranteed to be differentiable). 
These issues can be resolved, and in fact one can do so in the more general setting of semisimple orthogonal symmetric Lie algebras, see~\cite{Diag22} and in particular~\cite[Ex.~1.1 \& 1.2]{Reduced23}.
We will recall and adapt the pertinent results as necessary. 

The reduction of control systems was proven in detail in~\cite{Reduced23} in the setting of semisimple orthogonal symmetric Lie algebras.
In order to rigorously prove the Equivalence Theorem~\ref{thm:equivalence}, we need to show how the control systems considered here can be interpreted as control systems in such symmetric Lie algebras.
First we need the following generalization of the equivalence results proven in~\cite{Reduced23}.

\begin{proposition} \label{prop:equivalence}
Let $V$ be an $n$-dimensional real inner product space, let $\mb L$ be a compact Lie group with Lie algebra $\mf l$ acting on $V$, and let $Y$ be a linear vector field on $V$. 
Consider the control system
\begin{align} \label{eq:bilinear-gen}
\dot v = \big(Y+\sum_{i=1}^m u_i(t) l_i\big) v
\end{align}
with fast and full control on $\mb L$.
Moreover assume that we have a semisimple orthogonal symmetric Lie algebra $\mf g=\mf k\oplus\mf p$ with associated pair $(\mb G,\mb K)$ and maximal Abelian subspace $\mf a$.
Let $\imath:V\to\mf p$ be a linear isometric isomorphism and $\jmath : \mb L \to \Ad_{\mb K}$ a surjective Lie group homomorphism such that
\begin{align} \label{eq:ij-relation}
\jmath(L)\imath(v)=\imath(Lv).
\end{align}
Then the control system is equivalent in the sense of Theorem~\ref{thm:equivalence} to the following reduced control system on $W:=\imath^{-1}(\mf a)$:
\begin{align} \label{eq:reduced-gen}
\dot w(t) = Y_{L(t)}(w(t))
\end{align}
where $Y_L=\Pi_W\circ L^\star(Y)\circ\iota$, with $\Pi_W:V\to W$ the orthogonal projection onto $W$, $\iota:W\to V$ the inclusion, and $(\cdot)^\star$ the pullback. 
\end{proposition}

\begin{proof}
Let $\jmath_\star:=D\jmath(e)$ is the surjective Lie algebra homomorphism $\mf l\to\ad_{\mf k}$ corresponding to $\jmath$.
By surjectivity there are $k_i\in\mf k$ such that $\jmath_\star(l_i)=\ad_{k_i}$ for all $i=1,\ldots,m$.
By differentiating~\eqref{eq:ij-relation} we get $j_\star(l_i)\imath(v)=\imath(l_i(v))$.
Let $X:=\imath_\star(Y)$ be the drift vector field on $\mf p$.
If $\jmath(L)=\Ad_K$, then
$$
\imath^\star(X_K)
=
\imath^{-1}X_K\imath
=
\imath^{-1}\Pi_{\mf a}\Ad_K^{-1}X\Ad_K\imath
=
\Pi_W\imath^{-1}\Ad_K^{-1}X\Ad_K\imath
=
\Pi_W(\imath L)^{-1} X \imath L
=
Y_L.
$$
The remainder of the proof is split into two parts, the projection and the lift.

We begin with the projection.
Let $v:[0,T]\to V$ be a solution to~\eqref{eq:bilinear-gen}.
We want to show that $v^\down:[0,T]\to W$, defined by $\imath(v^\down)=\imath(v)^\down$, is a solution to~\eqref{eq:reduced-gen}.
We get that
$$
\frac{d}{dt}\imath(v)
=
\imath(\dot v)
=
\imath((Y+\sum_{i=1}^m u_i l_i) v)
=
\big(X+\sum_{i=1}^m\ad_{k_i}\big) \imath(v)
$$
almost everywhere. 
Hence $\imath(v)$ is a solution of the corresponding control system on $\mf p$, and we may apply~\cite[Thm.~3.8]{Reduced23} to obtain that $\imath(v)^\down$ is a solution of the reduced control system on $\mf a$, more explicitly, for almost every $t\in[0,T]$ there is some $K\in\mb K$ such that $\frac{d}{dt}\imath(v^\down(t))=X_K(\imath(v^\down(t)))$.
Next we show that $v^\down$ solves~\eqref{eq:reduced-gen}.
Indeed for the same $t$ we obtain by linearity of $\imath$ that $\dot v^\down=\imath^\star(X_K)(v^\down)=Y_L v^\down$ for some $L\in\mb L$.
This concludes the projection part of the proof.

Conversely, assume that we have a solution $w:[0,T]\to W$ to the reduced system~\eqref{eq:reduced-gen}.
Again we find for almost all $t\in[0,T]$ some $K\in\mb K$ such that $\frac{d}{dt}\imath(w(t))=\imath(\dot w(t))=X_K(\imath(w(t)))$,
and so $a:=\imath(w)$ solves the corresponding control system on $\mf a$. 
Hence there exists a corresponding control function $K:[0,T]\to\mb K$ which is measurable.
Using~\cite[Thm.~3.14]{Reduced23} we find approximate lifted solutions $p_\varepsilon:[0,T]\to\mf p$ with $\|p_\varepsilon-\Ad_Ka\|_{\infty}\leq\infty$.
As above we can compute with $v_\varepsilon:=\imath^{-1}(p_\varepsilon)$
$$
\frac{d}{dt}v_\varepsilon
=
\imath^{-1}(\dot p_\varepsilon)
=
\imath^{-1}((X+\sum_{i=1}^m u_i \ad_{k_i})p_\varepsilon)
=
(\imath^\star(X) +\sum_{i=1}^m u_i  l_i) v_\varepsilon
$$
and see that $v_\varepsilon$ is a solution to~\eqref{eq:bilinear-gen}.
Since $\imath$ is an isometry, for any measurable lift $L$ of $\Ad_K$ along $\jmath$ (which exists due to~\cite[Lem.~2.29]{Diag22}) we get
$$
\|v_\varepsilon-Lw\|_\infty
=
\|\imath(v_\varepsilon-Lw)\|_\infty
=
\|\imath(v_\varepsilon)-\jmath(L)\imath(w)\|_\infty
=
\|p_\varepsilon-\Ad_K(a)\|_\infty
\leq \varepsilon,
$$
it is also an $\varepsilon$-approximation.
This concludes the proof.
\end{proof}

See Table~\ref{tab:notation} for an overview of the notation related to the different control systems. 

\begin{sidewaystable}
\centering
\def\arraystretch{1.3}
\caption{We give a compact overview of the notation and the mathematical objects describing the various control systems presented in this paper, as well as the control systems of~\cite{Reduced23,LindbladReduced23}.}
\vspace{-2mm}
\label{tab:notation}
\begin{tabular}{llllll}
\hline\hline\\[-4mm]
 & General & Distinguishable & Bosonic & Fermionic & Lindbladian \\[2mm]
\hline\hline\\[-4mm]
Full System & \eqref{eq:bilinear-gen} & \eqref{eq:bilinear} & \eqref{eq:bilinear2} & \eqref{eq:bilinear2} & --- \\
Ambient space & $\mf p$ & $\C^{d_1}\otimes\C^{d_2}$ & $\mathrm{Sym}^2(\C^{d})$ & $\bigwedge^2(\C^{d})$ & $\mf{herm}_1(n)$ \\
State space & $S$ & $S(\C^{d_1}\otimes\C^{d_2})$ & $S(\mathrm{Sym}^2(\C^{d}))$ & $S(\bigwedge^2(\C^{d}))$ & $\mf{pos}_1(n)$ \\
State & $p$ & $\ket\psi$ & $\ket\psi$ & $\ket\psi$ & $\rho$ \\
Fast Controls & $\mb K$ & $\U_{\mathrm{loc}}(d_1,d_2)$ & $\U^s_{\mathrm{loc}}(d)$ & $\U^s_{\mathrm{loc}}(d)$ & $\SU(n)$ \\
Drift & $X$ & $-\iu H_0\in\mf u(d_1d_2)$ & $-\iu H_0\in\mf u^s(d^2)$ & $-\iu H_0\in\mf u^s(d^2)$ & $-L\in\mf{w}_{\mathsf{KL}}(n)$ \\
$\ad$ action & $\ad_a$ & $\ad^d_\sigma$ & $\ad^s_\sigma$ & $\ad^a_\xi$ & $\ad_\lambda$ 
\\[2mm]
\hline\hline\\[-4mm]
Reduced System & \eqref{eq:reduced-gen} & \eqref{eq:reduced} & \eqref{eq:reduced-sym} & \eqref{eq:reduced-asym} & --- \\
Ambient space & $\mf a$ & $\Sigma\cong\mf{diag}(d_1,d_2)$ & $\Sigma\cong\mf{diag}(d)$ & $\Xi\cong\mf{qdiag}(d)$ & $\Lambda=\mf{diag}_1(n)$ \\
State space & $R$ & $S^{d_{\min}-1}$ & $S^{d-1}$ & $S^{\floor{d/2}-1}$ & $\Delta^{n-1}$ \\
State & $a$ & $\sigma$ & $\sigma$ & $\xi$ & $\lambda$ \\
Weyl group & $\mb W$ & $S_{d_{\min}}\wr\Z_2$ & $S_d\wr\Z_2$ & $S_{\floor{d/2}}\wr\Z_2$ & $S_n$ \\
Inclusion & $\iota$ & $\diag$ & $\diag$ & $\qdiag$ & $\diag$ \\
Projection & $\Pi_{\mf a}$ & $\Pi_\Sigma$ & $\Pi_\Sigma$ & $\Pi_\Xi$ & $\Pi_{\diag}$ \\
Quotient map & $\pi$ & $\mathrm{sing}^\plusdown$ & $\mathrm{sing}^\plusdown$ & $\mathrm{qsing}^\plusdown$ & $\mathrm{spec}^\down$ \\
Induced vector fields & $X_K\in\mf X$ & $-H_{V\otimes W}\in\mf H$ & $-H^s_{V\otimes V}\in\mf H^s$ & $-H^a_{V\otimes V}\in\mf H^a$ & $-L_U\in\mf L$ \\
Natural wedge & --- & $\mf{so}(\dmin,\R)$ & $\mf{so}(d,\R)$ & $\mf{so}(\floor{d/2},\R)$ & $\mf{stoch}(n)$
\\[2mm]
\hline\hline\\[-4mm]
Decomposition & --- & Complex SVD & Autonne--Takagi fact. & Hua fact. & Hermitian EVD \\
Lie Algebra & --- & AIII & CI & DIII & A \\[2mm]
\hline\hline
\end{tabular}
\end{sidewaystable}

\subsection{Complex Singular Value Decomposition (Type AIII)} \label{app:AIII}

The complex singular value decomposition is encoded by the symmetric Lie algebra of type AIII, see for instance \cite[App.~A.7]{Diss-Kleinsteuber} and~\cite[Ch.~X \S2.3]{Helgason78}.
The standard matrix representation of this Lie algebra is the indefinite special unitary Lie algebra $\mf g_\AIII=\mf{su}(d_1,d_2)$ with Cartan decomposition $\mf g_\AIII=\mf k_\AIII\oplus\mf p_\AIII$ where\footnote{Often one denotes $\mf k_\AIII=\mf{s}(\mf u(d_1)\oplus\mf{u}(d_2))$.}
\begin{align*}
\mf k_\AIII &= \Big\{  \left(\begin{matrix}\iu E&0\\0&\iu F\end{matrix}\right) : \iu E\in\mf u(d_1),\, \iu F\in\mf u(d_2),\, \tr(E)=-\tr(F) \Big\},
\\
\mf p_\AIII &= \Big\{ \left(\begin{matrix}0&\psi\\\psi^*&0\end{matrix}\right) : \psi\in\C^{d_1,d_2} \Big\} \,.
\end{align*}
A choice of corresponding compact Lie group is $\mb K_\AIII=\mathrm S(\U(d_1)\times\U(d_2))$.

\begin{remark} \label{rmk:inner-prods}
In general we consider (semi)simple \emph{orthogonal} symmetric Lie algebras $\mf g=\mf k\oplus\mf p$ and so we also have to provide a ``compatible'' inner product on $\mf g$.
The inner product is (up to some irrelevant scaling) uniquely defined using the Killing form on $\mf g$, cf.~\cite[Ch.~V, Thm.~1.1]{Helgason78}.
Due to simplicity of $\mf g$ the Killing form is (again up to scaling) given by $\tr(AB)$.
In the following we will set the inner product on $\mf k$ to $-\tfrac12\tr(AB)$ and on $\mf p$ to $+\tfrac12\tr(AB)$, and $\mf k$ and $\mf p$ are orthogonal to each other.
Furthermore, we always use 
the real inner product $\Re(\braket{\psi,\phi})=\Re(\tr(\psi^*\phi))$ on states in $\C^{d_1}\otimes\C^{d_2}$.
\end{remark}

The spaces $\C^{d_1}\otimes\C^{d_2}$ and $\mf p_\AIII$ are identified using the map
$$
\imath^d : \C^{d_1}\otimes\C^{d_2} \to \mf p_\AIII ,\quad \ket\psi\mapsto\left(\begin{matrix}0&\psi\\\psi^*&0\end{matrix}\right).
$$ 

\begin{lemma} \label{lemma:imath-aiii}
The map $\imath^d$ is an $\R$-linear\footnote{Note that in the Lie algebraic context we always work with real vector spaces, even if their standard representation involves complex numbers.} isometric isomorphism.
The subspace $\mf a_\AIII:=\imath^d(\Sigma)$ is maximal Abelian and $\imath^d\circ\diag\circ\Pi_\Sigma = \Pi_{\mf a_\AIII}\circ\imath^d$.
The Weyl group $\mb W_\AIII$ is isomorphic to the generalized permutations $S_\dmin\wr\Z_2$ and $\mf w_\AIII:=\imath^d(\Sigma^\plusdown)$ is a Weyl chamber.
\end{lemma}

\begin{proof}
It is clear that $\imath^d$ is an $\R$-linear isomorphism.
With the inner product on $\C^{d_1}\otimes\C^{d_2}$ and $\mf p$ defined as in Remark~\ref{rmk:inner-prods} a simple computation shows that $\imath^d$ is even an isometry:
$$
\tfrac12\tr(\imath^d(\ket\psi),\imath^d(\ket\phi)) 
=
\tfrac12\tr(\psi\phi^* + \psi^*\phi) 
=
\Re(\tr(\psi^*\phi)).
$$
That $\imath^d(\Sigma)$ is maximal Abelian is well-known, cf.~\cite[Tab.~2]{Diag22}.
The fact that $\imath^d$ is an isometry also proves that $\imath^d\circ\diag\circ\Pi_\Sigma = \Pi_{\mf a_\AIII}\circ\imath^d$.
That the Weyl group acts by generalized permutations follows from the fact that the singular values are unique up to order and sign and the fact that any generalized permutation can be implemented by choosing $V$ and $W$ appropriately.
\end{proof}

Moreover we define the following maps:
\begin{alignat*}{4}
\jmath^d&:         \Uloc(d_1,d_2)\to\Ad_{\mb K_\AIII}, &\quad V\otimes W                          &\mapsto \Ad_{V\times\overline W} \\
\jmath^d_\star&:  \uloc(d_1,d_2)\to\ad_{\mf k_\AIII}, &\quad \iu E\otimes\id + \id\otimes\iu F &\mapsto \ad_{\iu E\times\overline{\iu F}}.
\end{alignat*}
Note that $\jmath^d_\star$ is the derivative of $\jmath^d$ at the identity.

\begin{lemma} \label{lemma:jmath-aiii}
It holds that $\jmath^d$ is a Lie group isomorphism, and so $\jmath^d_\star$ is a Lie algebra isomorphism\footnote{Contrary to $\imath^d$, the map $\jmath^d_\star$ is not an isometry with respect to the inner products of Remark~\ref{rmk:inner-prods}.}.
For $U\in\Uloc(d_1,d_2)$, $\iu H\in\uloc(d_1,d_2)$ and $\ket\psi\in\C^{d_1}\otimes\C^{d_2}$, the isomorphism $\imath^d$ and $\jmath^d$ satisfy the compatibility conditions
\begin{equation} \label{eq:compat-aiii}
\jmath^d(U)\imath^d(\ket\psi)=\imath^d(U\ket\psi), \quad \jmath^d_\star(\iu H)\imath^d(\ket\psi)=\imath^d(\iu H\ket\psi).
\end{equation}
Similarly we have the correspondence of the infinitesimal action $\imath^d(\ad_\sigma(\iu H))=-\jmath^d_\star(\iu H)(\imath^d(\ket{\sigma}))$ and of the induced vector fields where if $-X:=\imath^d_\star(\iu H_0)=\imath^d\circ(\iu H_0)\circ(\imath^d)^{-1}$ then $-H_U=(\imath^d\circ\diag)^\star X_{\jmath^d(U)}$.
Moreover the map $\mathrm{sing}^\plusdown$ corresponds to the quotient map with image in the Weyl chamber $\mf p_\AIII\to\mf a_\AIII/\mb W_\AIII\cong \mf w_\AIII$, see~\cite[Sec.~2.1]{Diag22}.
\end{lemma}

\begin{proof}
Even though $V\times\overline W$ does not always lie in $\mb K_\AIII$, we can choose $\phi\in\R$ such that $e^{\iu\phi}V\times\overline{e^{-\iu\phi}W}\in\mb K_\AIII$, and this phase disappears in the tensor product and in the adjoint representation.
Indeed any $\phi$ satisfying $e^{\iu\phi(d_1+d_2)}=\det W/\det V$ will do.
Hence $\jmath^d$ is well defined, and one easily verifies that it is an isomorphism.
The compatibility condition~\eqref{eq:compat-aiii} follows from a simple computation using the isomorphism of Section~\ref{sec:matrix-decs}.
Similarly the corresponding Lie algebra isomorphism $\jmath^d_\star$ can be written as $\iu E\otimes\id + \id\otimes\iu F \mapsto \ad_{\iu(E+\frac{\tr(F)-\tr(E)}{d_1+d_2}\id)\oplus\iu(-F+\frac{\tr(F)-\tr(E)}{d_1+d_2}\id)}$ to show explicitly that it is well defined.
By definition $\ad_\sigma(\iu H)=-\iu H\ket\sigma$. Using the compatibility condition~\eqref{eq:compat-aiii} this immediately yields $\imath^d(\ad_\sigma(\iu H))=-\jmath^d_\star(\iu H)(\imath^d(\ket{\sigma}))$ as desired.
Recall that we defined $H_U=\Pi_\Sigma\circ(U^*(\iu H_0)U)\circ\diag$ and in~\cite[Sec.~2.1]{Reduced23} we defined $X_K=\Pi_{\mf a}\Ad_K^\star(X)\circ\iota$. (We are slightly abusing notation here by writing $X_{\Ad_K}$.)
Using the compatibility condition~\eqref{eq:compat-aiii}, and $\Pi_{\mf a}\circ\imath^d=\imath^d\circ\diag\circ\Pi_\Sigma$, the claim $H_U=(\imath^d\circ\diag)^\star X_{\jmath^d(U)}$ follows from a simple computation.
The claim about the quotient map is just a restatement of the uniqueness of the singular values.
\end{proof}

\begin{remark} \label{rmk:svd-relation}
Explicitly~\eqref{eq:compat-aiii} states that $\jmath^d(V\otimes\overline W)\imath^d(\ket\psi)=\imath^d(\ket{V\psi W^*})$.
In a semisimple orthogonal symmetric Lie algebra every element in $\mf p_\AIII$ can be mapped into $\mf a_\AIII$ using the group action of $\mb K_\AIII$, cf.~\cite[Lem.~A.26]{Diag22}, which one might call ``diagonalization''.
In our case this means that $\imath^d(\ket\psi)$ can be mapped to some element in $\imath^d(\ket{V\psi W^*})\in\imath^d(\Sigma)$.
This exactly corresponds to the complex singular value decomposition.
Note however that in the Lie algebra setting we have $V\times W\in\mathrm S(\U(d_1)\times\U(d_2))$ and hence there is an additional restriction on the determinants of $V$ and $W$.
\end{remark}

\begin{remark} \label{rmk:ad-inv-def}
For regular $\sigma\in S^{\dmin-1}$ we have defined the map $\ad_\sigma^d$ in the main text and Lemma~\ref{lemma:jmath-aiii} shows that it is related to the adjoint representation $\ad_\sigma:\mf k_\AIII\to\mf p_\AIII$ (hence the name).
Denoting by $\mf k_\sigma^\perp$ and $\mf p_\sigma^\perp$ the orthogonal complement of the commutant of $\sigma$ in $\mf k_\AIII$ and $\mf p_\AIII$ respectively, it turns out that the restriction $\ad_\sigma:\mf k_\sigma^\perp\to\mf p_\sigma^\perp$ becomes bijective and hence invertible, see~\cite[Prop.~2.14]{Diag22}.
In fact this inverse is simply the Moore--Penrose pseudoinverse.
Moreover it holds that $\mf p_\sigma^\perp$ is just the orthocomplement of $\mf a$.
Hence, the pseudoinverse $(\ad_\sigma^d)^{-1}$ is defined on $\Sigma^\perp$ with image in $(\jmath^d_\star)^{-1}(k_\sigma^\perp)$.
Note however that $\jmath^d$ is not an isometry and so the orthocomplement has to be calculated in $\mf k$.
The precise definition is given in the following lemma.
\end{remark}

\begin{lemma} \label{lemma:ad-d-inv}
Let $\sigma\in S^{\dmin-1}$ be regular. 
Then it holds that
$\Sigma^\perp=\{ \ket\psi\in\C^{d_1}\otimes\C^{d_2} : \psi_{ii}\in\iu\R, 1\leq i\leq \dmin \}$,  
and the map $(\ad_\sigma^d)^{-1}$ can be described explicitly as
$$
(\ad_\sigma^d)^{-1} : \Sigma^\perp \to \uloc(d_1,d_2), \quad
\ket A \mapsto \iu E\otimes\id + \id\otimes\iu F\,,
$$
where $E_{ij}=0$ and $F_{ij}=0$ whenever $i>\dmin$ and $j>\dmin$, and for $i\leq\dmin$ and $j\leq\dmin$ we get
\begin{equation} \label{eq:ad-d-inv}
\iu E_{ii} = \iu F_{ii} = -\frac{A_{ii}}{2\sigma_i},
\quad
\iu E_{ij} = \frac{\sigma_jA_{ij} + \sigma_i\overline A_{ji}}{\sigma_i^2-\sigma_j^2},
\quad
\iu F_{ij} = \frac{\sigma_jA_{ji} + \sigma_i\overline A_{ij}}{\sigma_i^2-\sigma_j^2}.
\end{equation}
If $d_1>d_2$ resp. $d_1<d_2$ we additionally have
$$
\iu E_{ij} = \begin{cases} -\frac{A_{ij}}{\sigma_j} &\text{ if } j\leq d_2<i \\ 0 &\text{ if } j>d_2 \end{cases}
\quad\text{ resp. }\quad
\iu F_{ij} = \begin{cases} -\frac{A_{ji}}{\sigma_j} &\text{ if } i\leq d_1<j \\ 0 &\text{ if } i>d_1. \end{cases}
$$
\end{lemma}

\begin{proof}
We only consider the case $d_1\geq d_2$ for simplicity.
Consider the map $\iu E\otimes\id+\id\otimes\iu F\mapsto A:=-\iu(E\tilde\sigma+\tilde\sigma\overline F)$ where $\tilde\sigma$ is the corresponding diagonal matrix.
Then $A_{ii}=-\iu(E_{ii}+F_{ii})\sigma_i\in\iu\R$.
Hence, when we invert the map above, we will assume that $A_{ii}\in\iu\R$.
Moreover we need to find the kernel of the map, i.e.\ solve for $A=0$ (for all or any regular $\sigma$).
This happens if $E_{ii}+F_{ii}=0$, and $E_{ij}=F_{ij}=0$ for $j\leq d_2$ and $E_{ij}\in\C$ for $j>d_2$.
The orthocomplement of the kernel is given by $E_{ii}=F_{ii}$, and $E_{ij},F_{ij}\in\C$ for $j\leq d_2$ and $E_{ij}=0$ for $j>d_2$.
Using, for $i,j\leq\dmin$, that
$$
A_{ij}=-\iu(E_{ij}\sigma_j+F_{ji}\sigma_i), \quad
\overline A_{ji} = \iu(E_{ij}\sigma_i+F_{ji}\sigma_j),
$$
we find
\begin{align*}
\sigma_iA_{ij}+\sigma_j\overline A_{ji} &= \iu F_{ji}(\sigma_j^2-\sigma_i^2) \text{ thus } \iu F_{ij}=\frac{\sigma_jA_{ji} + \sigma_i\overline A_{ij}}{\sigma_i^2-\sigma_j^2} \\
\sigma_jA_{ij}+\sigma_i\overline A_{ji} &= \iu E_{ij}(\sigma_i^2-\sigma_j^2) \text{ thus } \iu E_{ij}=\frac{\sigma_jA_{ij} + \sigma_i\overline A_{ji}}{\sigma_i^2-\sigma_j^2}
\end{align*}
Finally for $j\leq d_2<i$ we find $\iu E_{ij}=-A_{ij}/\sigma_j$ and for $j>d_2$ we get $\iu E_{ij}=0$.
\end{proof}
\noindent Note that this lemma uniquely defines $\iu E\otimes\id + \id\otimes\iu F\in \uloc(d_1,d_2)$, although there is some freedom in the choice of $E$ and $F$ since we can shift some real multiple of the identity between them.

\subsection{Autonne--Takagi Factorization (Type CI)} \label{app:CI}

First discovered by Autonne~\cite{Autonne15} and Takagi~\cite{Takagi25}, the Autonne--Takagi factorization~\cite[Sec.~4.4]{HJ1ed2} states that for any complex symmetric matrix $A\in\mf{sym}(d,\C)$ there exists a unitary matrix $U\in\U(d)$ such that $UAU^\top$ is real and diagonal. The diagonal elements are uniquely defined up to order and signs, and they are in fact the singular values of $A$. 

The corresponding symmetric Lie algebra is that of type CI, usually represented by the real symplectic Lie algebra $\mf g_\CI=\mf{sp}(d,\R)$, see~\cite[Sec.~4.3]{Diss-Kleinsteuber} and again~\cite[Ch.~X \S2.3]{Helgason78}. 
The Cartan decomposition $\mf g_\CI=\mf k_\CI\oplus\mf p_\CI$ is given explicitly by
\begin{align*}
\mf k_\CI &= \left\lbrace\begin{bmatrix}
A&B\\-B&A
\end{bmatrix} : A=-A^\top, B=B^\top,\,\, A,B\in\R^{d,d}\right\rbrace,
\\
\mf p_\CI &= \left\lbrace\begin{bmatrix}
C&D\\D&-C
\end{bmatrix} : C=C^\top, D=D^\top,\,\, C,D\in\R^{d,d}\right\rbrace.
\end{align*}

\noindent The corresponding state space isomorphism $\imath^s$ is given by 
$$
\imath^s:\mathrm{Sym}^2(\C^d)\to\mf p_\CI,\, \ket\psi\mapsto\begin{pmatrix}\Re\psi&-\Im\psi\\-\Im\psi&-\Re\psi\end{pmatrix}.
$$

\begin{lemma} \label{lemma:imath-ci}
The map $\imath^s$ is an $\R$-linear isometric isomorphism.
The subspace $\mf a_\CI:=\imath^s(\Sigma)$ is maximal Abelian and $\imath^s\circ\diag\circ\Pi_\Sigma=\Pi_{\mf a_\CI}\circ\imath^s$.
The Weyl group $\mb W_\CI$ is isomorphic to the generalized permutations $S_d\wr\Z_2$ and $\mf w_\CI:=\imath^s(\Sigma^\plusdown)$ is a Weyl chamber.
\end{lemma}

\begin{proof}
The proof is analogous to that of Lemma~\ref{lemma:imath-aiii} so we just compute the inner product on $\mf p_\CI$:
$$
\frac12\tr(\imath^s(\ket\psi)\imath^s(\ket\phi))
=
\tr(\Re(\psi)\Re(\phi)+\Im(\psi)\Im(\phi))
=
\Re(\tr(\psi^*\phi)).
$$
Hence $\imath^s$ is an isometry and this concludes the proof.
\end{proof}
Now consider the following maps
\begin{alignat*}{4}
\jmath^s&:\Uloc^s(d)\to\Ad_{\mb K_{\CI}},\, &V\otimes V&\mapsto\Ad_{\begin{psmallmatrix}\Re V&\Im V\\-\Im V&\Re V\end{psmallmatrix}}, \\
\jmath^s_\star&:\uloc^s(d)\to\ad_{\mf k_{\CI}},\, &\iu H\otimes\id + \id\otimes\iu H &\mapsto \ad_{\begin{psmallmatrix}\Re(\iu H)&\Im(\iu H)\\-\Im(\iu H)&\Re(\iu H)\end{psmallmatrix}}.
\end{alignat*}

\begin{lemma} \label{lemma:jmath-ci}
The maps $\jmath^s$ and $\jmath^s_\star$ are Lie isomorphisms satisfying the compatibility conditions
\begin{equation} \label{eq:compat-ci}
\jmath^s(U)\imath^s(\ket\psi)=\imath^s(U\ket\psi), \quad \jmath^s_\star(\iu H)\imath^s(\ket\psi)=\imath^s(\iu H\ket\psi).
\end{equation}
As in Lemma~\ref{lemma:jmath-aiii} we get the correspondence of infinitesimal action, induced vector fields and quotient map.
\end{lemma}

\begin{proof}
Since $V\otimes V=(-V)\otimes(-V)$ we have to check that the map is well defined.
But it is clear that $\Ad_{-U}=\Ad_U$ and hence $\jmath^s$ is well defined.
That it is an isomorphism follows from~\cite[Prop.~4.7]{Diss-Kleinsteuber}.
The remainder of the proof is analogous to that of Lemma~\ref{lemma:jmath-aiii}.
\end{proof}

As in Remark~\ref{rmk:svd-relation}, the relation to the Autonne--Takagi factorization can be seen from~\eqref{eq:compat-ci}, which explicitly states that $\jmath^s(V\otimes V)\imath^s(\ket\psi)=\imath^s(\ket{V\psi V^\top})$, and from the fact that $\imath^s(\Sigma)=\mf a_\CI$.

As described in Remark~\ref{rmk:ad-inv-def} we can explicitly compute the appropriate inverse of the map $\ad^s_\sigma$.
Note that in this case we have $\mf k_\sigma^\perp = \mf k_\CI$.

\begin{lemma} \label{lemma:ad-s-inv}
Let $\sigma\in S^{d-1}$ be regular. Then it holds that $\Sigma^\perp=\{\ket\psi\in\mathrm{Sym}^2(\C^d) : \psi_{ii}\in\iu\R, 1\leq i\leq d\}$, and the map $(\ad_\sigma^s)^{-1}$ can be explicitly described as
$$
(\ad_\sigma^s)^{-1} : \Sigma^\perp \to \uloc^s(d), \quad \ket{A} \mapsto \iu E\otimes\id+\id\otimes\iu E,
$$
where 
$$
E_{ii}=\frac{\iu A_{ii}}2, \quad E_{ij} = -\frac{\Im(A_{ij})}{\sigma_i+\sigma_j} - \iu\frac{\Re(A_{ij})}{\sigma_i-\sigma_j}.
$$
\end{lemma}

\begin{proof}
Let $A=\ad_{\sigma}^s(\iu E) = -\iu (E\tilde\sigma + \tilde\sigma\overline E)$.
Then we have $A_{ii} = -2\iu E_{ii}\sigma_i$ and $A_{ij} = -\iu (E_{ij}\sigma_j+\sigma_i\overline E_{ij})$. 
Hence $E_{ii} = \frac{\iu A_{ii}}{2\sigma_i}$.
To invert the second equation we compute
\begin{align*}
A_{ij} = -\iu (E_{ij}\sigma_j+\sigma_i\overline E_{ij}), \quad
\overline A_{ij} = +\iu (\overline E_{ij}\sigma_j+\sigma_i E_{ij})
\end{align*}
and hence taking sum and difference we get
\begin{align*}
2\Re(A_{ij}) = \iu (E_{ij}-\overline E_{ij})(\sigma_i-\sigma_j) = -2\Im(E_{ij})(\sigma_i-\sigma_j)
&\Rightarrow \Im(E_{ij})=\frac{\Re(A_{ij})}{\sigma_j-\sigma_i}, \\
2\iu\Im(A_{ij}) = -\iu (E_{ij}+\overline E_{ij})(\sigma_i+\sigma_j) = -2\iu\Re(E_{ij})(\sigma_i+\sigma_j)
&\Rightarrow \Re(E_{ij}) = -\frac{\Im(A_{ij})}{\sigma_i+\sigma_j}.
\end{align*}
This concludes the proof.
\end{proof}

\subsection{Hua Factorization (Type DIII)} \label{app:DIII}

A skew-symmetric version of the Autonne--Takagi factorization also exists~\cite[Coro.~4.4.19]{HJ1ed2}.
It is called the Hua factorization, and was originally proven in~\cite[Thm.~7]{Hua44}. 
It states that for every skew-symmetric complex matrix $A\in\mf{asym}(d,\C)$ there exists a unitary $U\in\U(d)$ such that $UAU^\top$ is real and block diagonal with skew-symmetric blocks of size $2\times2$. 
If $d$ is odd then there is an additional $1\times1$ block containing a zero. 
We call such matrices quasi-diagonal.
Each $2\times2$ block is determined by a single real number (and its negative) which taken together yield the singular values of $A$.

This matrix factorization is related to the symmetric Lie algebra of type DIII, usually represented by the $\mf{so}^*(2d)$, see~\cite[App.~A.6]{Diss-Kleinsteuber} and again~\cite[Ch.~X \S2.3]{Helgason78}.\footnote{Note that~\cite{Helgason78} uses a different but isomorphic matrix representation.} In this case we have the following Cartan like decomposition
$$
\mf k_\DIII = \left\lbrace 
\begin{pmatrix}
\iu H & 0 \\ 0 &-\iu\overline H
\end{pmatrix} : \iu H\in\mf u(d) 
\right\rbrace,
\quad
\mf p_\DIII = \left\lbrace 
\begin{pmatrix}
0 & \psi \\ \psi^* & 0
\end{pmatrix} : \psi=-\psi^\top, \psi\in\C^{d,d}
\right\rbrace.
$$

\noindent Clearly the state space isomorphism is 
$$
\imath^a : \textstyle\bigwedge^2(\C^d)\to\mf p_\DIII, \quad \ket\psi \mapsto \begin{pmatrix}0&\psi\\\psi^*&0\end{pmatrix}.
$$

\begin{lemma} \label{lemma:imath-diii}
The map $\imath^a$ is an $\R$-linear isometric isomorphism. 
The subspace $\mf a_\DIII:=\imath^a(\Xi)$ is maximal Abelian and $\imath^a\circ\qdiag\circ\Pi_\Xi=\Pi_{\mf a}\circ\imath^a$.
The Weyl group $\mb W_\DIII$ is isomorphic to the generalized permutations $S_{\floor{d/2}}\wr\Z_2$ and $\mf w_\DIII:=\imath^a(\Xi^\plusdown)$ is a Weyl chamber.
\end{lemma}

\begin{proof}
The proof is entirely analogous to that of Lemma~\ref{lemma:imath-aiii}.
\end{proof}

The isomorphisms on the Lie group and algebra level are:

\begin{alignat*}{4}
\jmath^a &: \Uloc^s(d)\to\Ad_{\mb K_{\DIII}},& \quad V\otimes V &\mapsto \Ad_{\begin{psmallmatrix}V&0\\0&\overline V\end{psmallmatrix}} \\
\jmath^a_\star &: \uloc^s(d)\to\ad_{\mf k_{\DIII}},& \quad \iu H\otimes\id + \id\otimes\iu H &\mapsto \ad_{\begin{psmallmatrix}\iu H&0\\0&-\iu\overline H\end{psmallmatrix}}
\end{alignat*}

\begin{lemma} \label{lemma:jmath-diii}
The maps $\jmath^a$ and $\jmath^a_\star$ are Lie isomorphisms satisfying the compatibility conditions
\begin{equation} \label{eq:compat-diii}
\jmath^a(U)\imath^a(\ket\psi)=\imath^a(U\ket\psi), \quad \jmath^a_\star(\iu H)\imath^a(\ket\psi)=\imath^a(\iu H\ket\psi).
\end{equation}
As in Lemma~\ref{lemma:jmath-aiii} we get the correspondence of infinitesimal action, induced vector fields and quotient map.
\end{lemma}

\begin{proof}
Since $\Ad_{-U}=\Ad_U$ the map $\jmath^a$ is well-defined, and it is clearly a Lie group isomorphism.
The remainder of the proof is analogous to that of Lemma~\ref{lemma:jmath-aiii}.
\end{proof}

The relation to the Hua factorization can be seen form~\eqref{eq:compat-diii}, which becomes $\jmath^a(V\otimes V)\imath^a(\ket\psi)=\imath^a(\ket{V\psi V^\top})$, and the fact that $\imath^a(\Xi)=\mf a_\DIII$.


\begin{lemma} \label{lemma:ad-a-inv}
Let $\xi\in S^{\floor{d/2}-1}$ be regular.
It holds that $\Xi^\perp=\{\ket\psi\in\bigwedge^2(\C^d): \psi_{2i-1,2i}\in\iu\R, \, i=1,\dots,\floor{d/2} \}$.
The map $(\ad_\xi^a)^{-1}$ takes the following form:
$$
(\ad_\xi^a)^{-1} : \Xi^\perp \to \uloc^s(d), \quad \ket{A}\mapsto \iu E\otimes\id + \id\otimes\iu E,
$$
where for $1\leq i,j\leq \floor{d/2}$ we get
\begin{align*}
(\iu E)_{(ii)} = -\frac{a_i}{2\xi_i}\id_2, \quad 
(\iu E)_{(ij)} = \frac{\xi_i J \overline A_{(ij)} + \xi_j A_{(ij)}J}{\xi_j^2-\xi_i^2},
\end{align*}
where $(\iu E)_{(ij)}$ and $A_{(ij)}$ indexes the $2\times2$ blocks of the respective matrices, $\id_2=\begin{psmallmatrix}1&0\\0&1\end{psmallmatrix}$ and $J=\begin{psmallmatrix}0&1\\-1&0\end{psmallmatrix}$.
By $a_i\in\iu\R$ we denote the value satisfying $A_{ii}=a_iJ$.
In the case where $d$ is odd we additionally have
$$
\iu E_{d,2\iu-1}=\iu \overline E_{2\iu-1,d} = \frac{A_{d,2\iu}}{\xi_i}, \quad
\iu E_{d,2\iu}=\iu \overline E_{2\iu,d} = -\frac{A_{d,2\iu-1}}{\xi_i}, \quad
\iu E_{d,d}=0.
$$
\end{lemma}

\begin{proof}
First consider the even-dimensional case. 
Let $A:=\ad_\xi^a(\iu E\otimes\id + \id\otimes\iu E)=-\iu(E\tilde\xi+\tilde\xi\overline E)$.
For $1\leq i,j\leq \floor{d/2}$ we compute the bocks $A_{(ii)}=-\iu(E_{(ii)}\xi_iJ+\xi_iJ\overline E_{(ii)})$ as well as
\begin{align*}
A_{(ij)}= -\iu(E_{(ij)}\xi_jJ+\xi_iJ\overline E_{(ij)}), \quad \overline A_{(ij)}= +\iu(\overline E_{(ij)}\xi_iJ+\xi_jJE_{(ij)}).
\end{align*}
It follows that
$\xi_i J \overline A_{(ij)} + \xi_j A_{(ij)} J=\iu(\xi_j^2-\xi_i^2)E_{(ij)}$,
and hence we find that
$$
\iu E_{(ii)} = \frac{-a_i}{2\xi_i}\id_2, \quad 
\iu E_{(ij)} = \frac{\xi_i J \overline A_{(ij)} + \xi_j A_{(ij)} J}{\xi_j^2-\xi_i^2}.
$$
Here we used that the for $\jmath_\star(\iu E\otimes\id+\id\otimes\iu E)$ to lie in $\mf k_\xi^\perp$ the diagonal blocks $E_{(ii)}$ must be real multiples of the identity.
If $d$ is odd there is an additional column at the bottom and row on the right of $\iu E$ to be determined.
We find that
$$
A_{d,2i-1} = -\iu E_{d,2i}\xi_i, \quad
A_{d,2i} = +\iu E_{d,2i-1}\xi_i .
$$
The claimed results follow immediately.
\end{proof}

\bibliographystyle{acm}
\bibliography{./general.bib}

\end{document}